\documentclass{amsart}
\usepackage{amssymb, amsmath, color}
\usepackage[latin1]{inputenc}
\usepackage{graphicx}
\usepackage{float}
\usepackage{multirow}
\usepackage[ruled]{algorithm2e}
\usepackage{color}

\def\black{\color{black}}

\makeatletter
\@namedef{subjclassname@2020}{%
	\textup{2020} Mathematics Subject Classification}
\makeatother
\newtheorem{theorem}{Theorem}
\newtheorem{lemma}[theorem]{Lemma}
\newtheorem{remark}[theorem]{Remark}
\newtheorem{definition}[theorem]{Definition}
\newtheorem{corollary}[theorem]{Corollary}
\newtheorem{example}[theorem]{Example}

\begin{document}

\title[Combined quadrics]
{Contact detection between an ellipsoid and a combination of quadrics}
\author[Brozos \, Pereira-Sáez \, Rodríguez-Raposo \, Souto-Salorio \,  Tarrío-Tobar]{M. Brozos-V\'{a}zquez \, M.J. Pereira-Sáez \, A. B. Rodríguez-Raposo \, M.J. Souto-Salorio \, A. D. Tarrío-Tobar}
\address{MBV: Universidade da Coru\~na, Differential Geometry and its Applications Research Group,  Escola Polit\'ecnica Superior,  Spain}
\email{miguel.brozos.vazquez@udc.gal}
\address{MJPS: Universidade da Coru\~na, Differential Geometry and its Applications Research Group, Facultade de Econom\' ia e Empresa, Spain}
\email{maria.jose.pereira@udc.es}
\address{ABRR: Universidade de Santiago de Compostela, Departamento de Didácticas Específicas, Facultade de Ciencias da Educación, Spain}
\email{anabelen.rodriguez.raposo@usc.es}
\address{MJSS: Universidade da Coru\~na,  Differential Geometry and its Applications Research Group, Facultade de Inform\'atica, Spain}
\email{ maria.souto.salorio@udc.es}
\address{ADTT: Universidade da Coru\~na, Differential Geometry and its Applications Research Group, Escola Universitaria de Arquitectura Técnica, Spain}
\email{madorana@udc.es}
\thanks{Third and fourth authors partially supported by Agencia Estatal de Investigación (Spain), projects PID2020-115155GB-I00 and PID2020-113230RB-C21, respectively (European FEDER, UE)}

\begin{abstract}
We analyze the characteristic polynomial associated to an ellipsoid and another quadric in the context of the contact detection problem. We obtain a necessary and sufficient condition for an efficient method to detect contact. This condition is a feature on the size and the shape of the quadrics and can be checked directly from their parameters. Under this hypothesis, contact can be noticed by means of discriminants of the characteristic polynomial. Furthermore, relative positions can be classified through the sign of the coefficients of this polynomial.
\end{abstract}

\maketitle

\section{Introduction}\label{sect:introduction}

The contact detection problem between objects is recurrent in CAD/CAM. Many disciplines such as computer graphics, computer animation, robotics, industrial manufacturing or surgical simulation, among many others, require the detection of collisions between objects in many of their developments. Different surfaces have been used to model the great variety of shapes of the objects under consideration. Furthermore, depending on the chosen surfaces, appropriate methods based on their features are developed. For example, methods applied to polyhedra (see, for example, \cite{fengetal}) differ from others developed for differentiable surfaces (see \cite{Baraff,Ezair-Elber} and references therein). 

During last two decades, there has been an increasing use of quadric surfaces for modeling objects within the context of collision detection. This family of surfaces, together with conic curves, has been extensively studied, especially using techniques from Projective Geometry (see, for example, \cite{woods} for a classical reference where polarity is used to study the relative position of a pair of conics). When considering a pair of quadrics, much information about them is obtained using their associated pencil. Moreover, the characteristic polynomial obtained from the pencil provides important information linking the two surfaces. The work of Wang et al. \cite{wang-wang-kim} was seminal in introducing this polynomial associated to the pencil of two ellipsoids to detect contact between them. These methods have been extended to other quadric surfaces \cite{BPST_Hyperboloid,BV-PS-SS-TT} and exploited for practical uses, such as the detection of position for UAVs \cite{castro-et-al,dapena}. The analysis of the intersection of quadrics was initiated much earlier (see \cite{Levin79}) and continues to be an active research field (see \cite{Laureano2021,Pazouki-et-al,Wang-Goldman-Tu2003,Wang2004,Wang2009,wilf-manor} and references therein). 

Quadric surfaces allow to approximate very accurately a large variety of shapes. This is one of the main reasons for their use in contact detection problems. Also, since quadric surfaces are described by means of a quadratic polynomial, they are easier to handle than many other curved surfaces. In this paper we consider two quadrics, one of which is an ellipsoid. This particular surface is the only closed quadric surface. This feature makes it the most appropriate selection for modeling an object just by one surface. Moreover, the three degrees of freedom provided by its three axes allow to approximate many different objects. This justifies that an important part of the literature in this field involves ellipsoids (see, for instance, \cite{X2011,Pazouki-et-al,Wang2004,wang-wang-kim}). However, the ellipsoid has positive curvature and the shape of other objects requires the use of other quadric surfaces, for example, hyperboloids with negative curvature. In this paper we address the problem of contact detection between an ellipsoid and another quadric surface. 

Generally, we consider an ellipsoid $\mathcal{E}$ and another quadric surface $\mathcal{Q}$. While previous works as \cite{BPST_Hyperboloid,BV-PS-SS-TT,wang-wang-kim} treated particular quadrics, here we consider a wider class of surfaces. Along this work, the possible quadric surface $\mathcal{Q}$ is going to be one of the following: ellipsoid, hyperbolic or elliptic paraboloid, hyperboloid of one or two sheets, elliptic, parabolic or hyperbolic cylinder, or two planes. We shall make clear that we avoid the use of two coincidental planes as the quadric $\mathcal{Q}$, since from a geometric viewpoint they are equivalent to one plane. Let $E$ and $Q$ be their associated matrices. The {\it characteristic polynomial} of the pencil $\lambda E+Q$ is the fourth degree polynomial given by
\begin{equation}\label{eq:char-poly}
\mathfrak{P}(\lambda)=\det (\lambda E+Q).
\end{equation}
Notice that, since $E$ is non-degenerate, the roots of $\mathfrak{P}$ are the characteristic roots of the matrix $-Q E^{-1}$, so we will refer to them as the characteristic roots of $\mathfrak{P}$. 

The characteristic roots of $\mathfrak{P}$ permitted to detect the relative position between two ellipsoids, an ellipsoid and a paraboloid, or an ellipsoid and a hyperboloid of one sheet in \cite{wang-wang-kim,BPST_Hyperboloid,BV-PS-SS-TT} in some instances. 
In particular, it was shown that if there exists two complex conjugate (non-real) roots of $\mathfrak{P}$ then the quadric surfaces are in non-tangent contact. The converse is not true in general, as two quadrics may intersect non-tangentially and $\mathfrak{P}$ have four real roots (counted with multiplicity).

Since the existence of non-real roots can be easily detected by the discriminant of the polynomial $\mathfrak{P}$, it would be desirable to understand under which circumstances contact between quadrics can be noticed by a direct computation of the discriminant. This is the first aim of this work and with that purpose we introduce the following concept that relates the size and shape of the two quadric surfaces.
 
\begin{definition}\label{smallness-condition}
{\bf Smallness condition.} We say that the ellipsoid $\mathcal{E}$ is small with respect to the quadric surface $\mathcal{Q}$ if the intersection of the two quadric surfaces cannot be two curves at any relative position. 
\end{definition}

We consider that a curve is a $1$-dimensional connected set. 
The number of connected components of the intersection of two quadric surfaces ranges from $0$ to $2$ (see \cite{Degtyarev-et-al}), so the smallness condition rules out the possibility of two connected components in the intersection which are curves, but allows two isolated tangent points. 
We will see in Section~\ref{sect:proof-main-th} (Lemma~\ref{le:no-tangency-on-curve}) that the possibility of one  isolated tangent point and a curve is also eliminated by the smallness condition.
A similar definition was first given in \cite{BV-PS-SS-TT} to solve the particular problem of an ellipsoid and an elliptic paraboloid, although the condition was slightly more restrictive, since two tangent points were not allowed as a possible intersection set.

The smallness condition given in Definition~\ref{smallness-condition} is going to be analyzed in detail in Section~\ref{sect:smallness-condition}. Intersecting planes or cones do not satisfy Definition~\ref{smallness-condition} for any ellipsoid, whereas for other quadric surfaces it depends on some relations between the length axes and the curvature of the two surfaces. In Theorem~\ref{th:smallness-condition} we show how to check that $\mathcal{E}$ and $\mathcal{Q}$ satisfy the smallness condition by means of the parameters in the quadric equations. This characterization in terms of the parameters makes the condition more tractable computationally and allows a purely algorithmic checking. 

We respond to the first objective of this work showing that the smallness condition in Definition~\ref{smallness-condition} is a precise hypothesis that implies the equivalence between transversal contact (i.e., non-tangent contact) of the two quadric surfaces and the presence of non-real characteristic roots. 

\begin{theorem}\label{th:contact}
	Let $\mathcal{E}$ be a small ellipsoid with respect to the quadric surface $\mathcal{Q}$. Then $\mathcal{E}$ and $\mathcal{Q}$ are in transversal contact if and only if the characteristic polynomial $\mathfrak{P}(\lambda)$ has a pair of complex conjugate (non-real) roots. 
\end{theorem}

The proof of Theorem~\ref{th:contact} is given in Section~\ref{sect:proof-main-th}. The approach differs substantially to those followed in previous works as \cite{BPST_Hyperboloid,BV-PS-SS-TT}, since it is based on the analysis of the possible intersections between quadrics. Moreover,  the new approach relies on a combination of algebraic tools and methods from differential geometry. 

Based on Theorem~\ref{th:contact}, we can detect contact exclusively using discriminants associated to the characteristic polynomial, without the need of computing explicitly the roots of \eqref{eq:char-poly}. This provides an efficient way of detecting contact as will be shown in Section~\ref{sect:proof-main-th} (see Corollary~\ref{co:discriminant}). The detection of contact and relative positions between an ellipsoid and a plane is considered separately in Section~\ref{subsect:plane} (see Theorem~\ref{th:plane}).

Additionally to the detection of contact through the nature of the characteristic roots, the information encoded in the characteristic polynomial allows to detect the relative position between the two quadrics. Section~\ref{sect:relative-positions} is devoted to this task in the present context and previous results in \cite{BPST_Hyperboloid,BV-PS-SS-TT,wang-wang-kim} are extended. In Theorem~\ref{th:rel-pos}, the relative position of the small ellipsoid $\mathcal{E}$ with respect to the quadric surface $\mathcal{Q}$ is characterized in terms of the sign of the characteristic roots or, alternatively, in terms of the sign of the coefficients of $\mathfrak{P}(\lambda)$. Thus, this provides a computationally efficient method to approach the problem of identifying relative positions.

The second main goal of this work is to provide an efficient method to detect contact between a small ellipsoid and a combination of quadrics. The idea of composing geometric objects was considered, for example, in \cite{choi2014} for specific quadrics. We are considering here a more general context addressed in Section~\ref{sect:combined-quadrics}, where an algorithm is proposed for an efficient detection of the relative position between them. Thus, this proposal allows to model a great variety of real-world situations where two objects interact: a small ellipsoid models one of them and the other one is modeled by pieces of quadrics separated by a plane or other quadrics. A simple example is included to illustrate the method for particular quadric surfaces.

\section{The smallness condition}\label{sect:smallness-condition}
Along this section we analyze in detail the smallness condition given in Definition~\ref{smallness-condition}. First, we must emphasize that the smallness condition is a condition in a pair of surfaces and it depends on the relation between the two of them. Also, it is intrinsic to the geometry of the two surfaces, so it does not depend on a particular position, but on the possibility that the two surfaces intersect in two curves  when they are placed appropriately. As a consequence, since rigid motions do not alter the geometry of the surfaces, this smallness condition is invariant under rigid transformations of space. 

Notice that if the quadric $\mathcal{Q}$ is a pair of intersecting planes or a cone, then one can place $\mathcal{E}$ close enough to the intersecting ray or the vertex, respectively, to see that the smallness condition is not satisfied. Therefore, intersecting planes and cones are excluded from the analysis.

One of the main interests of Theorem~\ref{th:contact} is that we get a simple way to detect contact between the quadric surfaces. This provides an efficient algorithm with a simple implementation. Since the smallness condition is a necessary hypothesis, for practical purposes it would also be convenient to express it in a way that can be checked computationally. We will make this condition more tangible by inequalities in terms of the parameters of the quadrics.

 
Depending on the quadric $Q$ that we consider, the smallness condition in Definition~\ref{smallness-condition} results in different kind of restrictions. Some are related with the distance between particular points in $Q$ and affects directly to the axes of $\mathcal{E}$, whereas others depend on the curvature and impose conditions on the relations between the axes of $\mathcal{E}$. In order to specify them we consider a general ellipsoid $\mathcal{E}$ in standard form
\begin{equation}\label{eq:elipsoid-standar-form}
	\frac{x^2}{\alpha^2}+\frac{y^2}{\beta^2}+\frac{z^2}{\gamma^2}=1, \text{ with } \alpha\geq \beta \geq \gamma,
\end{equation}
and another quadric surface $\mathcal{Q}$.

\subsection{Direct restrictions on the axes (size).}\label{subse:restrictions-axes}
We assume $\mathcal{E}$ is small with respect to the quadric $\mathcal{Q}$. 
If $\mathcal{Q}$ is an ellipsoid, a hyperboloid of one or two sheets, an elliptic or a hyperbolic cylinder, or two parallel planes, then a first relation between the parameters of  $\mathcal{E}$ and $\mathcal{Q}$ is obtained by a direct study of distances between points of the surfaces.
If the quadric $\mathcal{Q}$ is an ellipsoid
\[
\frac{x^2}{a^2}+\frac{y^2}{b^2}+\frac{z^2}{c^2}=1, \text{ with } a\geq b \geq c,
\]
it is immediate that if the  ellipsoid $\mathcal{E}$ is small in comparison with $\mathcal{Q}$ then the largest axis of $\mathcal{E}$ must be smaller than or equal to the smallest axis of $\mathcal{Q}$. Hence, we conclude that $c\geq\alpha$. If the quadric $\mathcal{Q}$ is a hyperboloid of one sheet 
\[
\frac{x^2}{a^2}+\frac{y^2}{b^2}-\frac{z^2}{c^2}=1, \text{ with } a\geq b,
\]
then the major axis of $\mathcal{E}$ must be smaller than or equal to the smaller axis of $\mathcal{Q}$. Hence, we conclude that $b\geq\alpha$. Also, if the quadric $\mathcal{Q}$ is a hyperboloid of two sheets
\[
\frac{x^2}{a^2}+\frac{y^2}{b^2}-\frac{z^2}{c^2}=-1,\text{ with } a\geq b,
\]
then the major axis of $\mathcal{E}$ must be smaller than or equal to the distance between vertices in $\mathcal{Q}$. Hence we conclude that  $c\geq\alpha$.

The cases where $\mathcal{Q}$ is a cylinder can be projected orthogonally to a plane which is perpendicular to the axis. If $\mathcal{Q}$ is an elliptic cylinder $\frac{x^2}{a^2}+\frac{y^2}{b^2}=1$ with  $a\geq b$, it is clear that a necessary condition is $b\geq \alpha$. Whereas if $\mathcal{Q}$ is a hyperbolic cylinder  $\frac{x^2}{a^2}-\frac{y^2}{b^2}=1$ then $a\geq\alpha$.

Finally, for parallel planes with equation $\frac{x^2}{a^2}-1=0$, we observe that $2a$ is the minimum distance between two points that lie on different planes. Hence, the only restriction to avoid the possibility of two curves in the intersection is that $a\geq \alpha$.


\subsection{Restrictions on the relations between axes (shape).}\label{subse:conditions-curvature}
The previous conditions between the ellipsoid $\mathcal{E}$ and the quadric $\mathcal{Q}$, however, are not sufficient for the smallness condition to be satisfied. See, for example, Figure~\ref{fig:curv-cond-smallness-cond}. The curvature of the two quadric surfaces also plays a role in the verification of the smallness condition. Recall that the normal curvature $\kappa^S(p,\vec v)$ of a surface $S$ at a point $p$ in a fixed direction $\vec v$ is given by the curvature at $p$ of the curve obtained by the normal section in the direction of $\vec v$, this is, obtained by intersecting the surface with the normal plane at $p$ which contains $\vec v$. Also, the maximum and minimum normal curvatures  ($\kappa^S_{max}(p)$ and $\kappa_{min}^S(p)$) at a point $p$ are the principal curvatures of the surface at $p$. We denote by $\kappa_{max}^S$ and $\kappa_{min}^S$, respectively, the maximum and minimum principal curvature of the surface $S$. The other conditions that the surfaces $\mathcal{E}$ and $\mathcal{Q}$ must satisfy for $\mathcal{E}$ to be small in comparison with $\mathcal{Q}$ can be stated in terms of the principal curvatures of the surfaces.

\begin{figure}
	\includegraphics[height=3cm]{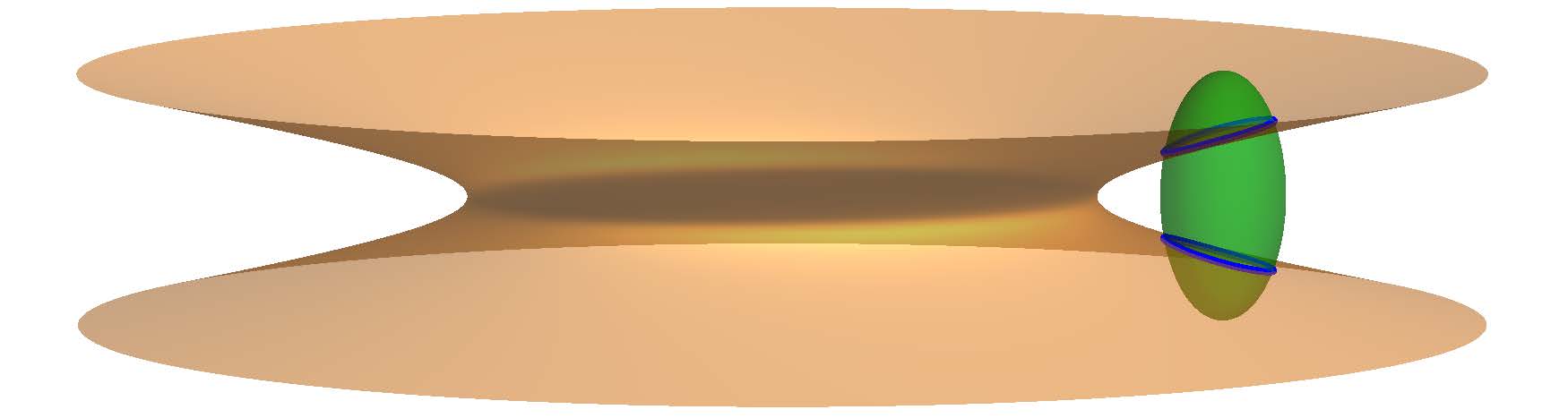}
	\caption{ The ellipsoid is not small with respect to the hyperboloid. There is a condition on the principal curvatures of the quadric surfaces for the smallness condition to be satisfied.}
	\label{fig:curv-cond-smallness-cond}
\end{figure}

\begin{lemma}\label{le:small-curvatures}
	If an ellipsoid $\mathcal{E}$ is small with respect to another quadric $\mathcal{Q}$ then 
	\[
	\kappa_{max}^\mathcal{Q}\leq \kappa_{min}^\mathcal{E}.
	\]
\end{lemma}
\begin{proof}
In order to compare the principal curvatures of the two surfaces, we are going to place the quadrics in the more favorable position for the presence of intersection curves. Then we reduce the problem in one dimension by considering sections by a suitable plane. Note that if the smallness condition is not satisfied, then there is an appropriate position between the quadrics so that the intersection has two curves. Hence an intentionally chosen normal section by a plane gives two conics that intersect in four different points. The conic obtained from the small ellipsoid is an small ellipse (i.e. an ellipse that cannot intersect the other conic in more than two points).
Consequently, the result follows directly from the following: 

{\bf Claim:} {\it if the smallness condition is satisfied, then the curvature at any point of the small ellipse is greater than or equal to the curvature at any point of the other conic}.

The next objective is to prove this claim. The problem trivializes if one conic is a ray, so we study the situations given by pairs of conics of the form ellipse-ellipse, ellipse-parabola or ellipse-hyperbola and analyze them separately as follows.

\noindent{\it Ellipse-ellipse}: in order to simplify the calculation and compare curvatures we place the small ellipse $E$ so that it is tangent to the other ellipse in one vertex (see Figure~\ref{fig:ellipse-conic}(i)). Now the corresponding equations are
	\begin{equation}\label{eq:elip-elip}
	E: \frac{(x-a+\beta)^2}{\beta^2}+\frac{y^2}{\alpha^2}=1\qquad \text{ and } \qquad C_1:\frac{x^2}{a^2}+\frac{y^2}{b^2}=1,
	\end{equation}
	where $\alpha\geq \beta>0$ and $a\geq b>0$. The diameter of $E$ is smaller than or equal to any axis of $C_1$ for the smallness condition to be satisfied. Hence $\alpha\leq b$ and we work with the inequalities in the parameters of the conics given by $a\geq b\geq \alpha\geq \beta>0$.
	\begin{figure}[h]
	\includegraphics[width=4cm]{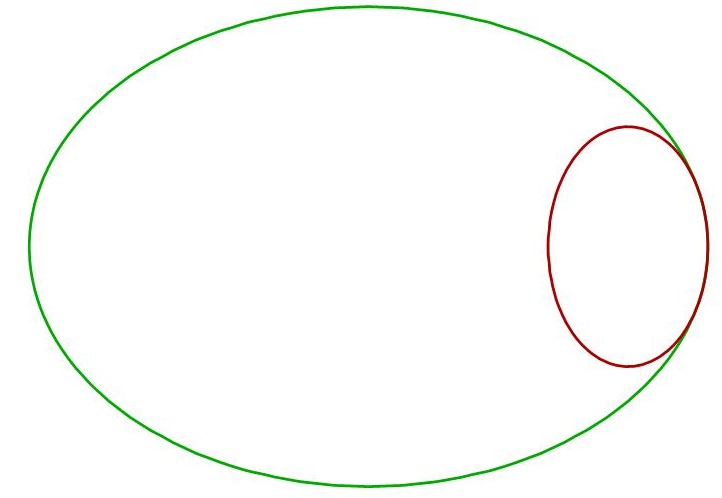} \includegraphics[width=4cm,height=2.8cm]{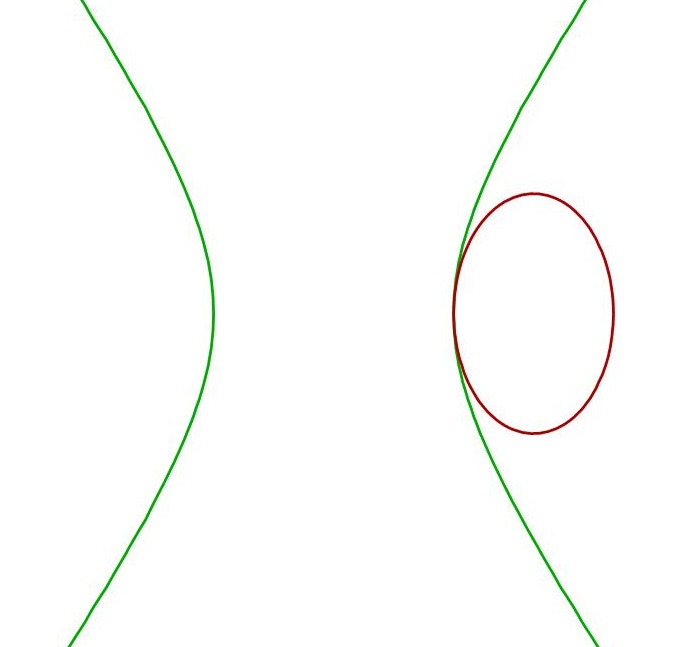}
	\includegraphics[width=4cm,height=2.8cm]{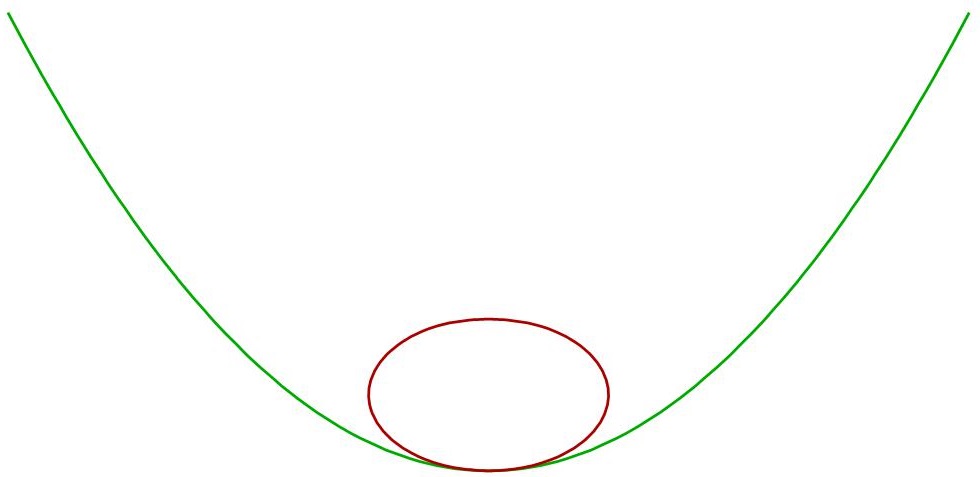}
	\qquad\qquad\qquad\qquad\qquad\qquad\qquad\qquad (i) \qquad\qquad\qquad\qquad\qquad (ii) \qquad\qquad\qquad\qquad\qquad  (iii)
	\caption{The small ellipse $E$ (red) is placed to be tangent to the conic $C$ at the point of minimum curvature of the $E$ and maximum curvature of $C$.}
	\label{fig:ellipse-conic}
	\end{figure}
Note that this placement of the pair of conics is obtained as a section from an appropriate placement of the pair of quadric surfaces. If $\beta=a$, then $\alpha=\beta=a=b$ and we have two circles of the same radius that can coincide in all points, which is an admissible situation. Henceforth  we assume $\beta < a$ and there is only one possible solution to the system of equations \eqref{eq:elip-elip} for the smallness condition to be satisfied. Note that if they intersect in three points, then a slight variation of the position makes them intersect at four different points. Using $y^2=b^2\left(1-\frac{x^2}{a^2}\right)$ we substitute in the  equation of $E$ to obtain
 \[
 x=a \text{ or } x=\frac{a \left(a^2 \alpha ^2-2 a \alpha ^2 \beta +b^2 \beta ^2\right)}{a^2 \alpha ^2-b^2 \beta ^2}.
 \]
The first solution corresponds to the tangent point $(a,0)$ and the second one gives
 \begin{equation}\label{eq:ellipse-ellipse}
 y^2 = -\frac{4 a \alpha ^2 b^2 \beta  (a-\beta ) \left(b^2 \beta -a \alpha ^2\right)}{(b \beta -a \alpha )^2
 	(a \alpha +b \beta )^2}.
 \end{equation}

For the smallness condition to be satisfied, there can not be more solutions for $y$ in \eqref{eq:ellipse-ellipse} than $y=0$. This implies that $y^2\leq 0$ in \eqref{eq:ellipse-ellipse}. Note that if $y=0$, then $b^2 \beta -a \alpha ^2=0$ and the expression for $x$ reduces to $x=a$. Now, if $y^2<0$, since $\beta < a$, all factors in \eqref{eq:ellipse-ellipse} have to be positive, so we conclude that $b^2 \beta -a \alpha ^2< 0$. In conclusion, we have that the desired relation between the parameters is $b^2 \beta -a \alpha ^2\geq  0$. 

An ellipse given by $C_1$ can be parameterized as $c_1(t)=(a\,\cos t,b\,sin t)$ with $t\in[0,2\pi]$. Since the curvature along the curve is given by $\kappa_1(t)=\frac{ab}{\sqrt{a^2\sin^2 t+b^2 \cos^2 t}^3}$ (see, for example, \cite{docarmo}), we have that the maximum curvature is attained at one vertex and its value is $\frac{a}{b^2}$. Analogously, we see that the minimum curvature of the small ellipse $E$ is $\frac{\beta}{\alpha^2}$. Thus, the relation $b^2 \beta -a \alpha ^2\geq 0$ can be written equivalently in terms of the maximum and minimum curvatures of the ellipses as 
\[
\frac{a}{b^2} \leq \frac{\beta}{\alpha^2}.
\]
 
\noindent{\it Ellipse-hyperbola:} we place the ellipse $E$, as in the previous case, tangent to the hyperbola at the vertex point (see Figure~\ref{fig:ellipse-conic}(ii)). The corresponding equations are
\begin{equation}\label{eq:hip-elip}
	E: \frac{(x-a-\beta)^2}{\beta^2}+\frac{y^2}{\alpha^2}=1\qquad \text{ and } \qquad C_2:\frac{x^2}{a^2}-\frac{y^2}{b^2}=1,
\end{equation}
where $\alpha\geq \beta$. An analogous argument to that given for two ellipses now shows the relation between the maximum curvature of the hyperbola, which is attained at the vertex point and values $\frac{a}{b^2}$, and the minimum curvature of the ellipse:
\[
\frac{a}{b^2}\leq \frac{\beta}{\alpha^2}.
\]

\noindent{\it Ellipse-parabola:} we place the small ellipse $E$ tangent at the vertex point of the parabola $C_3$ (see Figure~\ref{fig:ellipse-conic}(iii)) so that the equations of the conics are
\[
E: \frac{x^2}{\alpha^2}+\frac{(y-\beta)^2}{\beta^2}=1\qquad \text{ and } \qquad C_3:\frac{x^2}{a^2}-y=0.
\]
The smallness condition is satisfied if and only if the system of equations has only the solution $(0,0)$. We repeat the process above analyzing this two equations. As a result we obtain that the maximum curvature of the parabola, realized at the vertex point with value $\frac{2}{a^2}$, is less than or equal to the minimum curvature of the ellipse:
\[
\frac{2}{a^2}\leq \frac{\beta}{\alpha^2}.
\]
Hence the claim follows.
\end{proof}

\subsection{Smallness condition for $\mathcal{E}$ and $\mathcal{Q}$ in standard form}
Based on the previous analysis and the curvature conditions given in Lemma~\ref{le:small-curvatures}, we can now characterize the smallness condition in terms of relations between the parameters of the quadric surfaces. 

\begin{lemma}\label{le:small-cond-parameters-canonical-form}
Let $\mathcal{E}$ be an ellipsoid given by equation \eqref{eq:elipsoid-standar-form} and let $\mathcal{Q}$ be another quadric in standard form. 
Then the conditions  given in Table~\ref{table:smallness-condition-canonical-form} are necessary and sufficient  for $\mathcal{E}$ to be small with respect to $\mathcal{Q}$.

\begin{table}[H]
	\begin{tabular}{|c|c|c|@{\vrule height 10pt depth 8pt width 0pt} }
		\hline	
		
		{\bf Type of the quadric $\mathcal{Q}$} & {\bf Equation of $\mathcal{Q}$} & {\bf Conditions} \\
		
		\hline
		
		\multirow{2}{*}[-5pt]{Ellipsoid} & \multirow{2}{*}[-5pt]{$\frac{x^2}{a^2}+\frac{y^2}{b^2}+\frac{z^2}{c^2}=1$,\, $(a\geq b \geq c)$}& $c\geq \alpha$\\
		& & $\frac{a}{c^2}\leq \frac{\gamma}{\alpha^2}$ \\
		
		\hline
		
		\multirow{2}{*}[-5pt]{Hyperboloid of one sheet} & \multirow{2}{*}[-5pt]{$\frac{x^2}{a^2}+\frac{y^2}{b^2}-\frac{z^2}{c^2}=1$,\, $(a\geq b)$} & $b\geq\alpha$\\
		& &  $\frac{a}{c^2}\leq \frac{\gamma}{\alpha^2} $\\
		
		\hline
		
		\multirow{2}{*}[-5pt]{Hyperboloid of two sheets} & \multirow{2}{*}[-5pt]{$\frac{x^2}{a^2}+\frac{y^2}{b^2}-\frac{z^2}{c^2}=-1$,\, $(a\geq b)$} & $c\geq\alpha$\\
		& & $\frac{c}{a^2}\leq \frac{\gamma}{\alpha^2}$ \\
		
		\hline
		
		\raisebox{-3pt}{Elliptic paraboloid} & \raisebox{-3pt}{$\frac{x^2}{a^2}+\frac{y^2}{b^2}-z=0$,\, $(a\geq b)$} & \raisebox{-3pt}{$\frac{2}{b^2}\leq \frac{\gamma}{\alpha^2}$}  \\
		
		\hline
		
		\raisebox{-3pt}{Hyperbolic paraboloid} & \raisebox{-3pt}{$\frac{x^2}{a^2}-\frac{y^2}{b^2}-z=0$,\, $(a\geq b)$} &  \raisebox{-3pt}{$\frac{2}{b^2}\leq \frac{\gamma}{\alpha^2}$} \\
		
		\hline
		
		\multirow{2}{*}[-5pt]{Elliptic cylinder} & \multirow{2}{*}[-5pt]{$\frac{x^2}{a^2}+\frac{y^2}{b^2}=1$,\, $(a\geq b)$} & $b\geq \alpha$\\  
		& & $\frac{a}{b^2}\leq \frac{\gamma}{\alpha^2}$  \\
		
		\hline
		
		\multirow{2}{*}[-5pt]{Hyperbolic cylinder} & \multirow{2}{*}[-5pt]{$\frac{x^2}{a^2}-\frac{y^2}{b^2}=1$} & $a\geq\alpha$\\
		& &  $ \frac{a}{b^2}\leq \frac{\gamma}{\alpha^2}$ \\
		
		\hline
		
		\raisebox{-3pt}{Parabolic cylinder} & \raisebox{-3pt}{$\frac{x^2}{a^2}-z=0$} & \raisebox{-3pt}{$ \frac{2}{a^2}\leq \frac{\gamma}{\alpha^2}$} \\
		
		\hline
		
		\raisebox{-3pt}{Pair of parallel planes} & \raisebox{-3pt}{$\frac{x^2}{a^2} - 1 = 0$} & \raisebox{-3pt}{$a\geq \alpha$}\\
		
		\hline
	\end{tabular}
	\vspace{1em}
	\caption{Smallness condition in terms of the parameters of the ellipsoid $\mathcal{E}$ and the other quadric surface $\mathcal{Q}$ in standard form.}
	\label{table:smallness-condition-canonical-form}
\end{table}

\end{lemma}
\begin{proof}
In the case of the ellipsoid, the hyperboloids and the elliptic and hyperbolic cylinders, when doing sections by normal planes, one can obtain ellipses and hyperbolas. Thus, considering any section, the diameter of the ellipse given by $\mathcal{E}$ has to be smaller than the axes of the ellipse and the transverse axis of the hyperbola of $\mathcal{Q}$. These give rise to the first condition in each of these cases in Table~\ref{table:smallness-condition-canonical-form}. For all the quadric surfaces but parallel planes there is a condition in terms of the principal curvatures that was given in Lemma~\ref{le:small-curvatures}.
 
If the smallness condition is not satisfied, then there is a position where the two surfaces intersect in two curves. Moving the surface adequately if necessary, a case by case analysis of the pair of surfaces shows that  the condition on the axes of the quadrics  (see Subsection~\ref{subse:restrictions-axes}) or the condition on the curvature (see Subsection~\ref{subse:conditions-curvature}) is not satisfied.
\end{proof}

\subsection{Smallness condition for $\mathcal{E}$ and $\mathcal{Q}$ in general form}

The results given in this paper are intended to be used in practical real-life contexts, where one does not generally have quadric surfaces in standard form. Since the smallness condition is a necessary hypothesis, it is convenient to have a way to check it, without necessarily changing  coordinates and reducing equations to standard form. 
As the smallness condition is invariant under rigid transformations, the associated invariants to quadric surfaces are enough to determine whether the smallness condition holds for a given pair of quadrics. In this subsection we recall which are the needed invariants and provide the precise relations between them to check the smallness condition. 
	
The general equation of a quadric in Euclidean coordinates ${x_1,x_2,x_3}$ given by
\begin{equation}\label{eq:general-eq-quad}
	\sum_{i,j=1}^3 a_{ij} x_i x_j+\sum_{i=1}^3 2 b_i x_i+ c=0, \text{ where } a_{ij}=a_{ji},
\end{equation}
can be written as $X^T Q X = 0$ with $X^T = (x_1,x_2,x_3,1)$ and $Q$ the quadric's matrix:
\[
Q=\left(   \begin{array}{c@|c}
	a_{ij}	& b_j \\
	\hline
	b_i & c
\end{array} \right ) \text{ with } i,j=1,2,3.
\]
Associated to this equation, we have the following invariants:
\begin{itemize}
	\item The determinant of $Q$: $\det(Q)$.
	\item The eigenvalues of $Q_{00} = (a_{ij})$, that are labeled as $(EV):\mu_1, \mu_2$ and $\mu_3$. Observe that, as a consequence, the trace of $Q_{00}$, $\operatorname{tr}(Q_{00})=\mu_1+\mu_2+\mu_3$, and the determinant of $Q_{00}$, $\det (Q_{00})=\mu_1\mu_2\mu_3 $, are also invariant.
	\item $J = \det\left(   \begin{array}{rr}
		a_{11}&a_{12} \\
		a_{12}&a_{22}
	\end{array} \right ) +  \det\left(   \begin{array}{rr}
		a_{11}&a_{13} \\
		a_{13}&a_{33}
	\end{array} \right ) + \det\left(   \begin{array}{rr}
		a_{22}&a_{23} \\
		a_{23}&a_{33}
	\end{array} \right )$.
\item If $\tilde Q_{ij}$ is the adjoint matrix of $a_{ij}$ in $Q$: $K= \det \tilde Q_{11} + \det \tilde Q_{22}+ \det \tilde Q_{33} $.
\item $J' = \det\left(   \begin{array}{rr}
		a_{11} & b_1 \\
		b_1 & c
	\end{array} \right ) +  \det\left(   \begin{array}{rr}
		a_{22} & b_2 \\
		b_2 & c
	\end{array} \right ) + \det\left(   \begin{array}{rr}
		 a_{33} & b_3 \\
		b_3 &  c
	\end{array} \right )$.
\end{itemize}
%

\begin{theorem}\label{th:smallness-condition}
Let $\mathcal{E}$ be an ellipsoid and $\mathcal{Q}$  another quadric.
Then the relations given in Table~ \ref{table:smallness-condition-general} are necessary and sufficient  for $\mathcal{E}$ to be small with respect to $\mathcal{Q}$.
\end{theorem}
\begin{proof}
Due to the invariance under rigid transformations of the smallness conditions and the invariants associated to the quadric surfaces, the result follows from Lemma~\ref{le:small-cond-parameters-canonical-form}.
\end{proof}

\begin{table}[H]
	\begin{tabular}{|l|l|l|c|@{\vrule height 10pt depth 10pt width 0pt}}
		\hline
		\multicolumn{4}{|c|}{
			General ellipsoid $\mathcal{E}$ with reduced equation $\frac{x^2}{\alpha'^2} + \frac{y^2}{\beta'^2} + \frac{z^2}{\gamma'^2} = \delta'^2$}\rule[-0.4cm]{0cm}{1cm}\\
		\multicolumn{4}{|c|}
		{$EV: \frac{1}{\alpha'^2}, \frac{1}{\beta'^2}, \frac{1}{\gamma'^2}, \alpha'\geq\beta'\geq\gamma'$
			and  $\delta'^2 = -\alpha'^2\beta'^2\gamma'^2\det(E)$}\rule[-0.5cm]{0cm}{0.8cm}\\
		\hline\hline
		
		{\bf Quadric $\mathcal{Q}$} & {\bf Invariants} & {\bf Reduced equation} &  {\bf Conditions}\\
		
		\hline
		
		\multirow[c]{2}{*}[-0.5em]{Ellipsoid} & \raisebox{-2pt}{$EV: \frac{1}{a'^2}, \frac{1}{b'^2}, \frac{1}{c'^2},\;a'\geq b'\geq c'$} & \multirow[c]{2}{*}[-0.5em]{$\frac{x^2}{a'^2}+\frac{y^2}{b'^2}+\frac{z^2}{c'^2} = d'^2$} & $d'c' \geq \delta'\alpha'$\\		
		&	$d'^2 = -a'^2b'^2c'^2\det(Q)$ & & $\frac{a'}{d'c'^2}\leq \frac{\gamma'}{\delta'\alpha'^2}$\\
		
		\hline
		
		\multirow[c]{2}{2cm}[-0.5em]{Hyperboloid of one sheet} & \raisebox{-2pt}{$EV: \frac{1}{a'^2}$, $\frac{1}{b'^2}$, $-\frac{1}{c'^2}, a'\geq b'$} & \multirow[c]{2}{*}[-0.5em]{$\frac{x^2}{a'^2} + \frac{y^2}{b'^2} - \frac{z^2}{c'^2} = d'^2$} & $b'd'\geq \delta'\alpha'$\\
		& $d'^2 = a'^2b'^2c'^2\det(Q)$ & & $\frac{a'}{d'c'^2}\leq \frac{\gamma'}{\delta'\alpha'^2}$\\
		
		\hline
		
		\multirow{2}{2cm}[-0.5em]{Hyperboloid of two sheets} & \raisebox{-2pt}{$EV: \frac{1}{a'^2}$, $\frac{1}{b'^2}$, $-\frac{1}{c'^2}, a'\geq b'$} & \multirow[c]{2}{*}[-0.5em]{$\frac{x^2}{a'^2} + \frac{y^2}{b'^2} - \frac{z^2}{c'^2} = d'^2$} & $c'd'\geq \delta'\alpha'$\\		
		& $d'^2 = - a'^2b'^2c'^2\det(Q)$ & & $\frac{c'}{d'c'^2}\leq \frac{\gamma'}{\delta'\alpha'^2}$\\
		
		\hline
		
		\multirow[c]{2}{*}[-0.5em]{Elliptic paraboloid} & \raisebox{-2pt}{$EV: \frac{1}{a'^2}$, $\frac{1}{b'^2}$, $0, a'\geq b'$} & \multirow[c]{2}{*}[-0.5em]{$\frac{x^2}{a'^2} + \frac{y^2}{b'^2} - Lz = 0$} & \multirow[c]{2}{*}[-0.5em]{$\frac{2}{Lb'^2}\leq \frac{\gamma'}{\delta'\alpha'^2}$}\\		
		& $L^2=-\frac{4\det(Q)}{J}$ & & \\
		
		\hline
		
		\multirow[c]{2}{2cm}[-0.5em]{Hyperbolic paraboloid} & \raisebox{-2pt}{$EV: \frac{1}{a'^2}$, $-\frac{1}{b'^2}$, $0, a'\geq b'$} & \multirow[c]{2}{*}[-0.5em]{$\frac{x^2}{a'^2} - \frac{y^2}{b'^2} - Lz = 0$} & \multirow[c]{2}{*}[-0.5em]{$\frac{2}{Lb'^2}\leq \frac{\gamma'}{\delta'\alpha'^2}$}\\	
		& $L^2=-\frac{4\det(Q)}{J}$ & & \\
		
		\hline
		
		\multirow[c]{2}{*}[-0.5em]{Elliptic cylinder} &	\raisebox{-2pt}{$EV: \frac{1}{a'^2}$, $\frac{1}{b'^2}$, $0, a'\geq b'$} & \multirow[c]{2}{*}[-0.5em]{$\frac{x^2}{a'^2} + \frac{y^2}{b'^2} = M^2$} & $Mb'\geq \delta'\alpha'$ \\		
		& $M^2=-\frac{K}{J}$ & & $\frac{a'}{Mb'^2}\leq \frac{\gamma'}{\delta'\alpha'^2}$ \\
		
		\hline
		
		\multirow[c]{2}{*}[-0.5em]{Hyperbolic cylinder} & \raisebox{-2pt}{$EV: \frac{1}{a'^2}$, $-\frac{1}{b'^2}$, $0, a'\geq b'$} & \multirow[c]{2}{*}[-0.5em]{$\frac{x^2}{a'^2} - \frac{y^2}{b'^2} = M^2$} & $Mb'\geq \delta'\alpha'$ \\	
		& $M^2=-\frac{K}{J}$ & & $\frac{a'}{Mb'^2}\leq \frac{\gamma'}{\delta'\alpha'^2}$ \\
		
		\hline
		
		\multirow[c]{2}{*}[-0.5em]{Parabolic cylinder} & \raisebox{-2pt}{$EV: \frac{1}{a'^2}$, $0$, $0$} & \multirow[c]{2}{*}[-0.5em]{$\frac{x^2}{a'^2} -d'z=0$} & \multirow[c]{2}{*}[-0.5em]{$\frac{2}{d'a'^2}\leq \frac{\gamma'}{\delta'\alpha'^2}$} \\
		&  $d'^2=4Ka'^2$  & & \\
		
		\hline
		
		\multirow[c]{2}{*}[-0.5em]{Parallel planes} & \raisebox{-2pt}{$EV: \frac{1}{a'^2}$, $0$, $0$} &
		\multirow[c]{2}{*}[-0.5em]{$\frac{x^2}{a'^2} -d'^2=0$} & \multirow[c]{2}{*}[-0.5em]{$a'd'\geq \alpha'\delta'$}\\
		& $d'^2 = - a'^2 J'$ & &\\
		
		\hline
	\end{tabular}
	\vspace{0.5em}
	\caption{Conditions on the invariants to check that an ellipsoid $\mathcal{E}$ is small with respect to a quadric surface $\mathcal{Q}$}.
	\label{table:smallness-condition-general}
\end{table}

Note that, from the relations given in Table~ \ref{table:smallness-condition-general}, it is immediate to verify whether the smallness condition holds for a given pair of quadrics.

\section{Contact detection between the quadric surfaces}\label{sect:proof-main-th}

In this section we assume $\mathcal{E}$ is an small ellipsoid with respect to a quadric $\mathcal{Q}$. We deal with results aimed to detect contact between the two quadric surfaces. First we give the proof of Theorem~\ref{th:contact} and, secondly, we provide more efficient methods based on the use of a system of discriminants for the characteristic polynomial.

\subsection{Proof of Theorem~\ref{th:contact}}
We prove the main result by studying the possible intersections between two quadric surfaces. In order to do that, we begin by analyzing intersections which are a curve with a tangent point (see Figure~\ref{fig:curves-with-tangency}).

\begin{lemma}\label{le:common-curve}
	If $\mathcal{E}$ and $\mathcal{Q}$ intersect in a curve $C$ and there is a point $p\in C$ where the surfaces are tangent, then there is one normal curvature that coincides for the two surfaces at $p$.
\end{lemma}
\proof
Let $C$ be the intersection curve for $\mathcal{E}$ and $\mathcal{Q}$. Analyzing the possibilities of intersection curves between the quadric surfaces (see, for example, \cite{Wang2009,wilf-manor}) we see that $C$ is a differentiable curve except, perhaps, at the tangent point $p$ (for example, if the intersection curve is a cuspidal quartic then there is a singularity at the cusp at $p$). We work in a neighborhood of $p$ and parameterize an arc of $C$ as $\alpha:(a,b]\to \mathbb{R}^3$ so that $\alpha(b)=p$, $\alpha$ is continuous in $(a,b]$ and smooth in $(a,b)$.

Since $\alpha$ is smooth in $(a,b)$, we choose a regular parameterization and compute $c'(t)$ for all $t\in (a,b)$. We normalize $c'(t)$ to assign a unit vector $\vec v(t)=c'(t)/\|c'(t)\|$ to each point $\alpha(t)$ for $t\in (a,b)$. Now, we define $\vec v(b):=\lim_{t\to b} \vec v(t)$ to extend the unit tangent vector to $\alpha(b)$.

Let $\kappa^{\mathcal{E}}(t)$ and $\kappa^{\mathcal{Q}}(t)$ be the normal curvatures at $\alpha(t)$ in the direction of $\vec v(t)$ in $\mathcal{E}$ and $\mathcal{Q}$, respectively, for $t\in (a,b]$. These two functions $\kappa^{\mathcal{E}}$ and $\kappa^{\mathcal{Q}}$ vary smoothly along $\alpha$. Now, note that the normal curvatures $\kappa^{\mathcal{E}}(t)$ and $\kappa^{\mathcal{Q}}(t)$ can be obtained from the curvature of $\alpha(t)$ by projecting on the normal vector to each of the surfaces for $t\in (a,b)$. Observing that the normal vector of $\mathcal{E}$ and $\mathcal{Q}$ varies smoothly along $\alpha(t)$ for $t\in (a,b]$ and that the normal vector of the two surfaces coincides in $\alpha(b)$, by continuity of the normal vector, of $\vec v$, of   $\kappa^{\mathcal{E}}$ and of $\kappa^{\mathcal{Q}}$, we conclude that the normal curvature of the two surfaces in the direction of $\vec v$ at $p$ is the same.
\qed

\begin{figure}[h]
	\includegraphics[height=3cm]{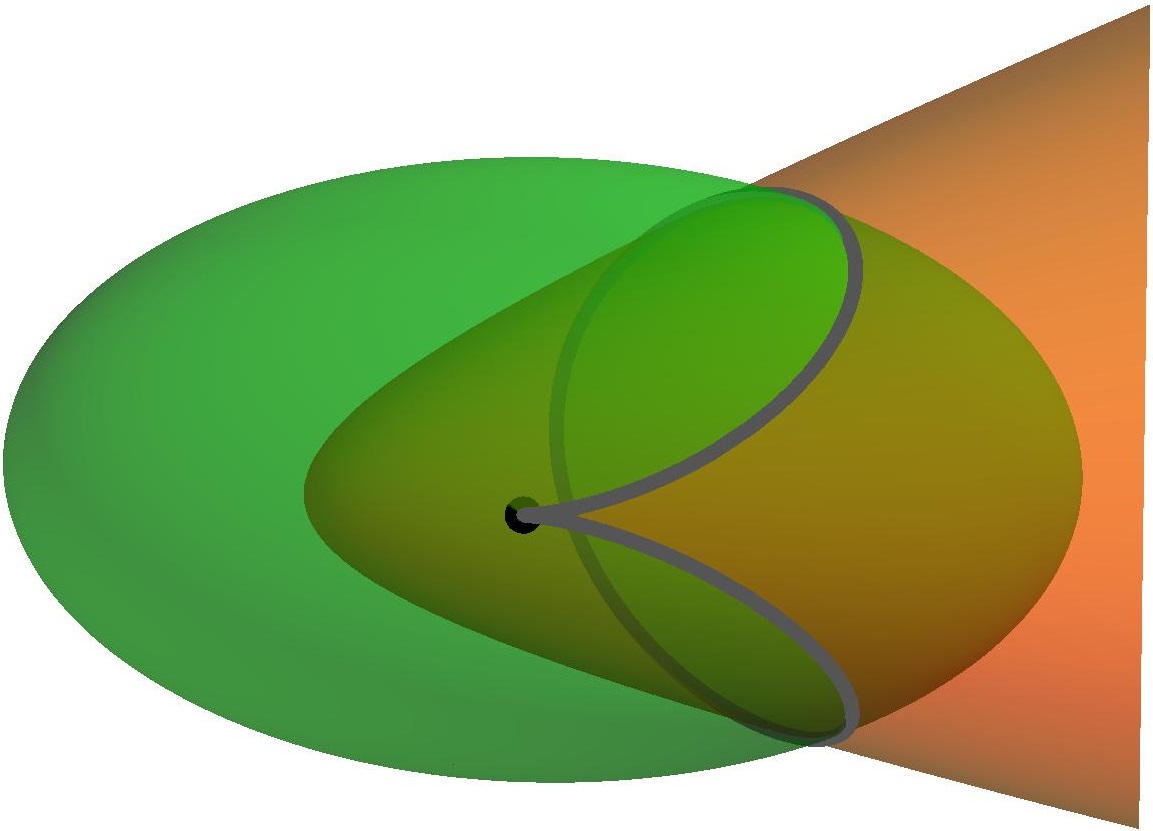}
	\includegraphics[height=3cm]{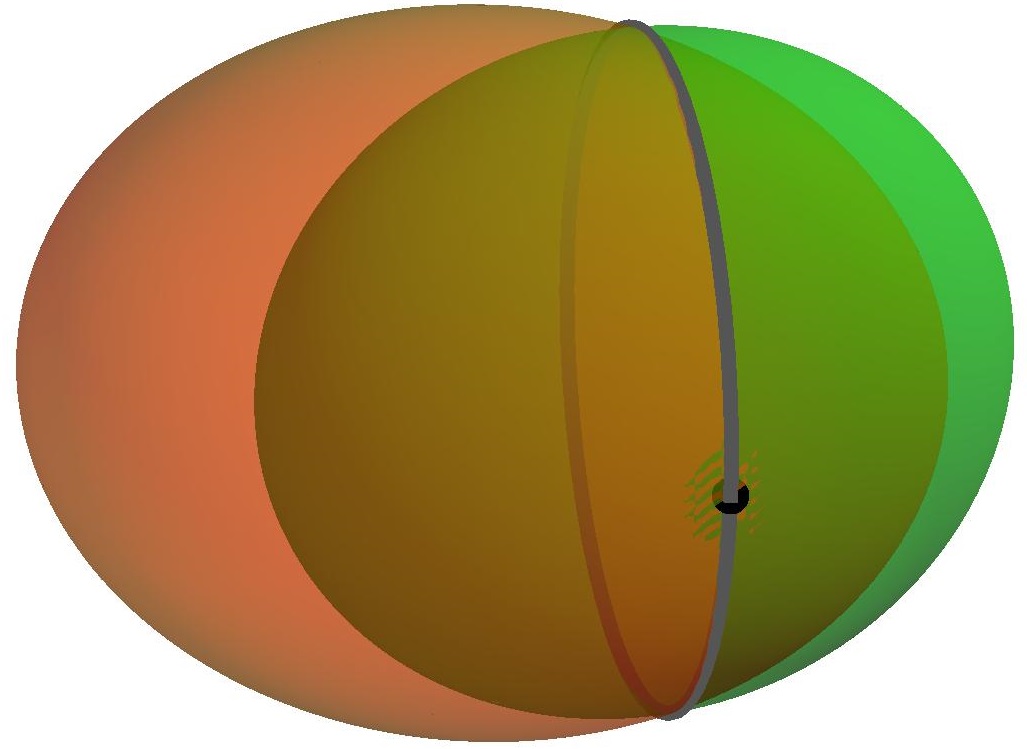}
	\caption{Intersection curves including one tangent point.}\label{fig:curves-with-tangency}
\end{figure}

\begin{lemma}\label{le:no-tangency-on-curve}
	Let $\mathcal{E}$ be small with respect to $\mathcal{Q}$. If they intersect in a curve $C$ and are tangent at a point $p$, then $C$ is a circle, $p$ belongs to $C$ and the quadric surfaces are tangent along $C$.
\end{lemma}
\begin{proof}
We begin by proving that $p$ belongs to $C$. We argue by contradiction, so we assume first that $p$ does not belong to the curve $C$. Then the intersection of the quadrics has one connected component which is the curve and another connected component which is the tangent point. A straightforward analysis of the morphology of the possible intersections between two quadrics shows that this tangent point is an isolated tangent point. Therefore, slightly translating $\mathcal{E}$ in the appropriate direction transforms the isolated tangent point in a differentiable curve. Thus, after the translation there are two connected components which are curves, so the smallness condition is not satisfied. Hence, we conclude that $p$ belongs to $C$.

Now, if $p$ belongs to $C$,  then Lemma~\ref{le:common-curve} applies and we have that
the normal curvature $\kappa (p)$ of the two surfaces at $p$ in the direction of the curve coincides. If $\kappa_{min}^\mathcal{E}$, $\kappa_{max}^\mathcal{E}$ and $\kappa_{min}^\mathcal{Q}$, $\kappa_{max}^\mathcal{Q}$ are the principal curvatures of the surfaces, in virtue of Lemma~\ref{le:small-curvatures}, the following relation is satisfied:
\[
\kappa_{min}^\mathcal{Q}\leq \kappa_{max}^\mathcal{Q}\leq \kappa_{min}^\mathcal{E}\leq \kappa_{max}^\mathcal{E}.
\]
Since $\kappa_{min}^\mathcal{Q}\leq \kappa(p)\leq \kappa_{max}^\mathcal{Q}$ and $\kappa_{min}^\mathcal{E}\leq \kappa(p)\leq \kappa_{max}^\mathcal{E}$,
we conclude that $\kappa_{max}^\mathcal{Q}=\kappa(p)= \kappa_{min}^\mathcal{E}$. Hence the normal curvature in the tangent direction of $C$ at $p$ is principal and, moreover, is the maximum principal curvature for $\mathcal{Q}$ and the minimum principal curvature for $\mathcal{E}$.
Therefore, the point $p$ has to be a vertex of $\mathcal{E}$ with minimum principal curvature and a vertex of $\mathcal{Q}$ with maximum principal curvature. Note that, because $\mathcal{E}$ and $\mathcal{Q}$ are tangent at $p$, the two surfaces share the same tangent plane at $p$. Now, we consider a section of $\mathcal{E}$ and $\mathcal{Q}$ by a plane trough $p$ which is orthogonal to the tangent plane at $p$ and that intersects $C$ at least in another point $q$ different from $p$ (this is possible, since $C$ is closed). This plane intersects $\mathcal{E}$ and $\mathcal{Q}$ in two conic curves: $c^\mathcal{E}$ (which is an ellipse) and $c^\mathcal{Q}$ (possibly with two connected components). The point $p$ is a vertex for the two curves and the curvature of $c^\mathcal{E}$ at $p$ is greater than or equal to the curvature of $c^\mathcal{Q}$ at $p$. Hence $c^\mathcal{Q}$ is necessarily an ellipse or has two connected components. More specifically, $\mathcal{Q}$ is an ellipsoid, a hyperboloid of one sheet or an elliptic cylinder. 

If $\mathcal{Q}$ is an ellipsoid, since $\mathcal{E}$ is an small ellipsoid which shares a vertex with $\mathcal{Q}$ where they are tangent, the only possibility is that the two ellipsoids are tangent along the greater ellipse of $\mathcal{E}$. But, because this ellipse is common to the two ellipsoids and the maximum normal curvature of $\mathcal{Q}$ is smaller than the minimum normal curvature of $\mathcal{E}$ (see Lemma~\ref{le:small-curvatures}), the curvature at all points of the ellipse must be the same, so it is a circle. Moreover, since the curvature of this circle is the minimum normal curvature for $\mathcal{E}$, this quadric has to be an sphere. As a consequence of the smallness condition, since the two ellipsoids share a tangent circle, the only possibility in this case is that the two ellipsoids are coincidental spheres. 

Similar arguments are used if $\mathcal{Q}$ is an hyperboloid or an elliptic cylinder to conclude that the tangent curve $C$ is a circle, but in this case we get admissible cases as in Figure~\ref{fig:curves-with-tangency2}.  
\qed

\begin{figure}[h]
	\begin{tabular}{c@{\hskip 2.2cm} c}
	\includegraphics[height=3cm]{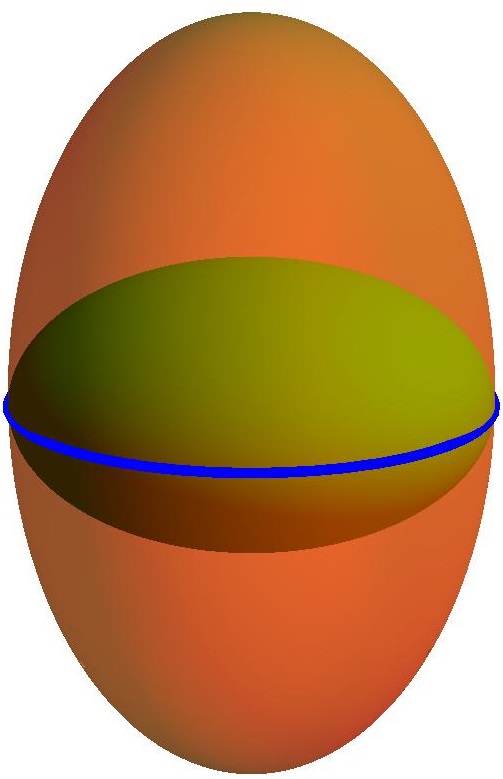}&
\includegraphics[height=3cm]{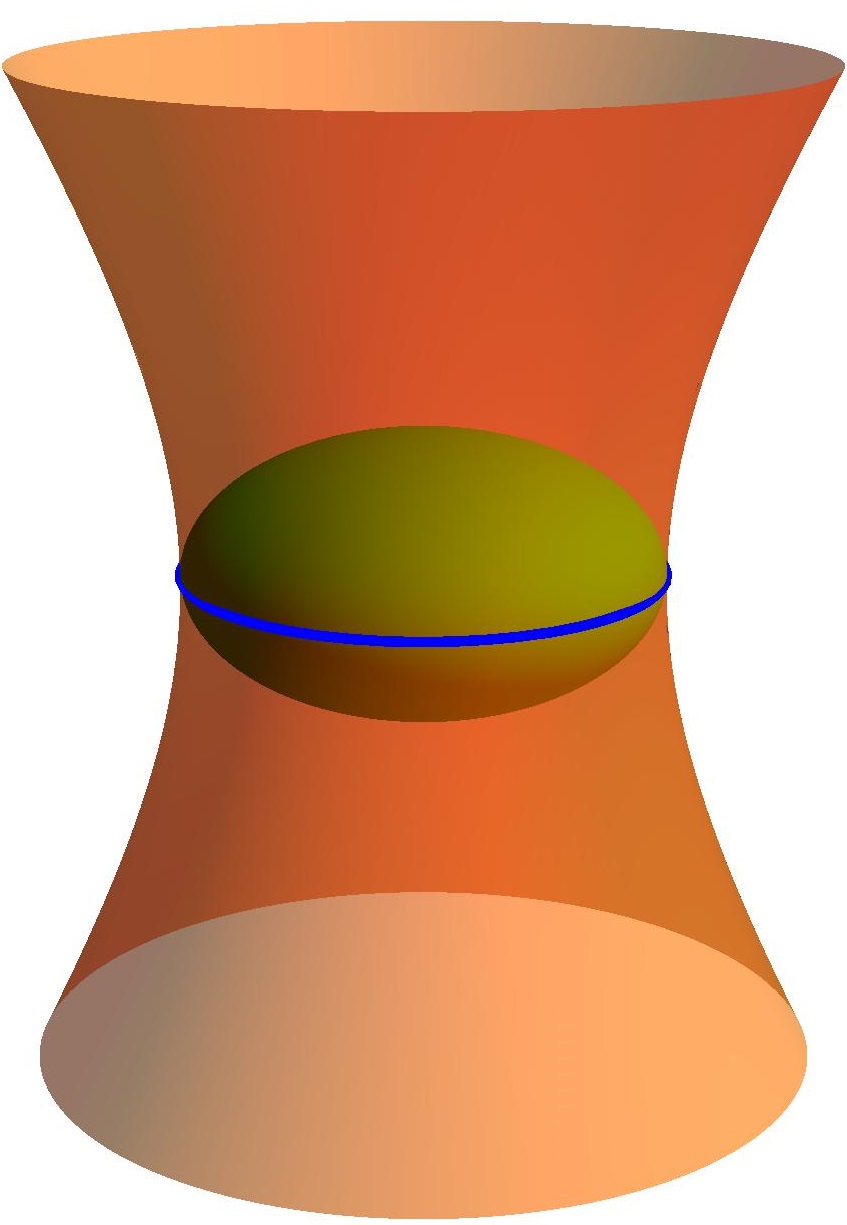}\qquad
	\includegraphics[height=3cm]{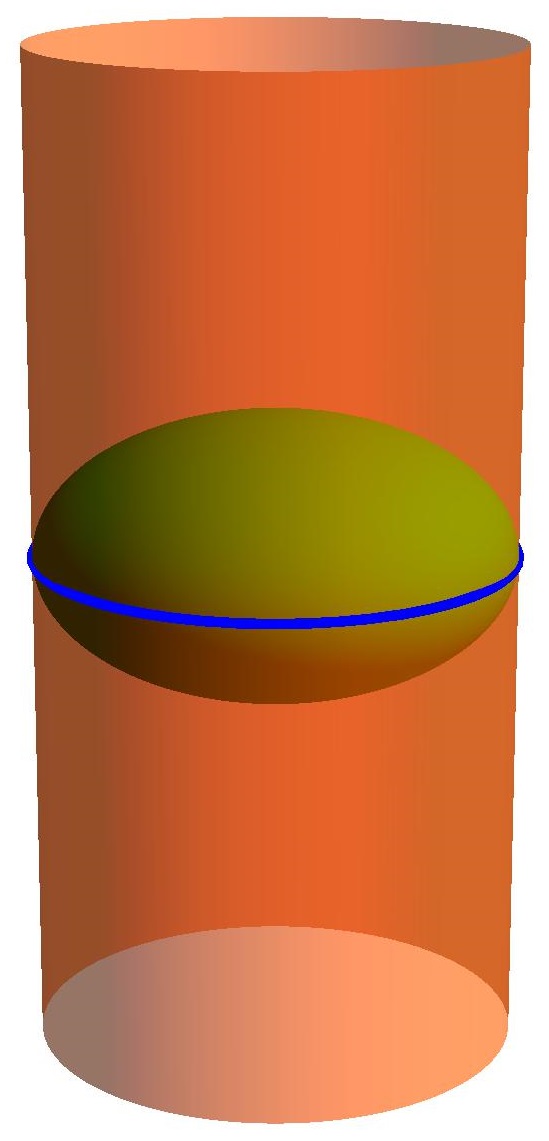}\\
	 (a) & (b)
	\end{tabular}
	\caption{(a)\, Two ellipsoids sharing a tangent circle do not satisfy the smallness condition. (b)\, A circular hyperboloid of one sheet and a circular cylinder provide admissible examples.}\label{fig:curves-with-tangency2}
\end{figure}

{\it Proof of Theorem~\ref{th:contact}.}
 
We consider an ellipsoid $\mathcal{E}$ and another quadric surface $\mathcal{Q}$ in the affine space. For convenience, the quadric surfaces can be thought at the real projective space $\mathbb{R}P^3$ and the affine space is a realization of it where we choose a plane at infinity. Since one of the quadric surfaces is an ellipsoid, the plane at infinity does not intersect the quadric surface and, therefore, the intersection of the two quadrics does not have points at the plane at infinity. Hence, when we consider the possible intersection curves in $\mathbb{R}P^3$, only homotopically null curves are admissible, which give rise to close curves in affine space. Since, moreover, we consider a ellipsoid which is small in comparison with the other quadric, we can also eliminate some other cases. Neither a curve and an isolated tangent point nor a curve with a tangent point are compatible with this hypothesis, as shown in Lemma~\ref{le:no-tangency-on-curve}. A direct application of the smallness condition also rules out the possibility of two curves as the intersection set.

From a topological point of view and attending to the classification of intersections between quadrics given in \cite{Wang2009}, we have the following possibilities for the intersection set of $\mathcal{E}$ and $\mathcal{Q}$, related with their respective Segre types and the roots of the characteristic polynomial:  
\begin{enumerate}
	\item {\it $\emptyset$:  no contact between the quadrics.} The possible Segre types are $[1111]_4$, $[(11)11]_3$, $[(111)1]_2$ and $[(11)(11)]_2$. Thus, the eigenvalues are real.
	\item {\it Isolated tangent points.} There are two possibilities, both of which can be realized:
	\begin{enumerate}
		\item $1$ isolated tangent point: cases 7, 22, 25 and 33 in \cite{Wang2009}. The Segre types are $[211]_3$, $[(21)1]_3$, $[2(11)]_2$ and $[(211)]_1$, which correspond to real eigenvalues.
		\item $2$ isolated tangent points: cases 15 and 30 in \cite{Wang2009}, with corresponding Segre types $[(11)11]_3$ and $[2(11)]_2$. As in the previous case, all eigenvalues are real. 
	\end{enumerate} 
	\item {\it One curve.} Depending on the existence of tangency, we consider two cases:
	\begin{enumerate}
		\item {\it $1$ connected component with no tangent points:} cases 3, 13 and 17 in \cite{Wang2009}, with corresponding Segre types $[1111]_2$ or $[(11)11]_1$. Hence, there is a pair of complex conjugate (non real) roots.
		\item {\it $1$ connected component with all points of tangency:} case 19 in \cite{Wang2009}.  The Segre type is $[(111)1]_2$, so there is a triple root and a single root. 
	\end{enumerate} 
\end{enumerate}
We include the previous classification in Table~\ref{table:eigenvalues-contact}. 
As a conclusion, under the hypotheses of Theorem~\ref{th:contact}, the transversal contact is identified by non-real roots. Hence, the characterization of Theorem~\ref{th:contact} follows.
\end{proof}

\begin{remark}\label{re:necessary-condition}\rm
We shall emphasize that the smallness condition given in Definition~\ref{smallness-condition} is a necessary hypothesis in Theorem~\ref{th:contact}. Indeed, if the intersection has two connected components which are two curves, then the associated Segre types are $[1111]_4$ and $[(11)11]_3$, with four real roots. Also, in the two cases where the intersection is a curve and an isolated tangent point the associated Segree types are $[211]_3$ and $[2(11)]_2$, again with real roots. Hence, the smallness condition cannot be relaxed in Theorem~\ref{th:contact}.
However, non-real roots for the characteristic polynomial always imply transversal contact, even if the smallness condition fails. Moreover, this is a general fact for any pair of quadrics. A direct analysis of the type of intersection between two quadrics (see \cite{Wang2009}) for Segre types with non-real roots shows the following:
\begin{quote}\it
	If the characteristic polynomial associated to any two quadrics has non-real roots, then they are in transversal contact.
\end{quote}
\end{remark}

\subsection{Contact detection using discriminants of $\mathfrak{P}(\lambda)$}
From a Complete Discrimination System one can determine the number and multiplicities of the real roots.
We consider a general characteristic polynomial \eqref{eq:char-poly}, which has degree four:
\begin{equation}\label{eq:char-poly-coef}
	\mathfrak{P}(\lambda)=c_4 \lambda^4+c_3\lambda^3 +c_2 \lambda^2+c_1 \lambda + c_0,
\end{equation}
where $c_0,\dots,c_4$ are coefficients determined by the parameters of the quadrics.
 In the case at hand, where we can detect transversal contact between quadrics just by checking two Segre types, we only need two terms of the discrimination system. We define (see \cite{emeris,yang}):
\[
\begin{array}{rcl} 
	\Delta_3 &=& 16 c_4^2c_0c_2-18 c_4^2 c_1^2-4c_4 c_2^3+14 c_4c_1c_3c_2-6c_4c_0c_3^2+c_2^2c_3^2-3c_1c_3^3,
		\\
	\noalign{\bigskip}
	\Delta_4 &=&    256 c_0^3 c_4^3-192 c_0^2 c_1 c_3 c_4^2-128 c_0^2 c_2^2 c_4^2+144 c_0^2 c_2 c_3^2 c_4-27 c_0^2 c_3^4\\
	\noalign{\medskip}
	&&
	+144 c_0 c_1^2 c_2 c_4^2-6 c_0 c_1^2 c_3^2
	c_4-4 c_1^3 c_3^3-80 c_0 c_1 c_2^2 c_3 c_4+18 c_0 c_1 c_2
	c_3^3
	\\
	\noalign{\medskip}
	&&
	+16 c_0 c_2^4 c_4-4 c_0 c_2^3 c_3^2	-27 c_1^4
	c_4^2+18 c_1^3 c_2 c_3 c_4-4 c_1^2
	c_2^3 c_4+c_1^2 c_2^2 c_3^2.
\end{array}
\]
Using the determination of roots in terms of these two discriminants, we obtain the following consequence of Theorem~\ref{th:contact}.
\begin{corollary}\label{co:discriminant}
	Let $\mathcal{E}$ be a small ellipsoid with respect to the quadric surface $\mathcal{Q}$. Then $\mathcal{E}$ and $\mathcal{Q}$ are in transversal contact if and only if one of the following holds:
	\begin{enumerate}
		\item $\Delta_4<0$,
		\item $\Delta_4=0$ and $\Delta_3<0$.
	\end{enumerate}
\end{corollary}
\begin{proof}
In the proof of Theorem~\ref{th:contact} we saw that the possible Segre types for the transversal contact are $[1111]_2$ or $[(11)11]_1$. Following \cite{yang}, the Segre type $[1111]_2$ is determined by $\Delta_4<0$ and  $[(11)11]_1$ by $\Delta_4=0$ and $\Delta_3<0$ (see Table~\ref{table:eigenvalues-contact}).
\end{proof}

\begin{table}[H]
	\centering
	\begin{tabular}{|p{3cm}|ll|c|l|@{\vrule height 10pt depth 8pt width 0pt}}
		\hline
		{\bf Type of contact} & \multicolumn{2}{c|}{\bf Segre type} & {\bf Roots} & {\makebox[2.5cm][c]{$\mathbf{\Delta_{3,4}}$}}  \\
		
		\hline
		
		\multirow[c]{2}{*}[-0.5em]{No contact} & $[1111]_4$ & $[(11)11]_3$ & \multirow[c]{2}{*}[-0.5em]{$\mathbb{R}$} & $\Delta_4> 0$ or\\	
		& $[(111)1]_2$ & $[(11)(11)]_2$ & & $\Delta_4=0 \wedge \Delta_3\geq 0$\\
		
		\hline
		
		\multirow[c]{2}{3cm}[-0.7em]{1 isolated	tangent point} & $[211]_3$ & $[(21)1]_3$ & \multirow[c]{2}{*}[-0.5em]{$\mathbb{R}$} & \multirow[c]{2}{*}[-0.7em]{$\Delta_4=0 \wedge \Delta_3\geq 0$}\\
		& $[2(11)]_2$ & $[(211)]_1$ & & \\
		
		\hline
		
		2 isolated tangent points & \raisebox{-4pt}{$[(11)11]_3$} & \raisebox{-4pt}{$[2(11)]_2$}  & \raisebox{-4pt}{$\mathbb{R}$} & \raisebox{-4pt}{$\Delta_4=0 \wedge \Delta_3\geq 0$}\\
		
		\hline
		
		\multirow[c]{2}{3cm}[-0.5em]{Curve with no tangent points} & \multirow[c]{2}{*}[-0.5em]{$[1111]_2$} & \multirow[c]{2}{*}[-0.5em]{$[(11)11]_1$} & \multirow[c]{2}{*}[-0.5em]{$\mathbb{C}$} & $\Delta_4< 0$ or\\
		& & & & $\Delta_4=0 \wedge \Delta_3< 0$\\
		
		\hline
		
		Curve of tangency & $[(111)1]_2$ & $[(11)11]_3$  & $\mathbb{R}$&  $\Delta_4=0 \wedge \Delta_3\geq 0$\\
		
		\hline
		
	\end{tabular}
	\medskip
	\caption{\small Relations between the type of contact and the characteristic roots: Segre type, real nature of the roots and discriminants.}
	\label{table:eigenvalues-contact}
\end{table}

The possibility (1) in Corollary~\ref{co:discriminant} is more likely to appear in real world applications than the possibility (2), since the later appears only with a double real root and the set of this configuration has zero measure in the total space. However it has to be taken into account for the implementation of a collision detection algorithm. The following example illustrates this phenomenon. 

\begin{example}
We consider the ellipsoid $2x^2+2y^2+3z^2-1=0$ and the elliptic paraboloid $x^2+y^2+8z=0$. The characteristic polynomial is $\mathfrak{P}(\lambda)=-(1+2\lambda)^2(16+3\lambda^2)$, so there is a real double root $-\frac12$ and complex conjugate roots $\pm \frac{4 i}{\sqrt{3}}$. The discriminants satisfy:
\[
\Delta_4=0 \text{ and } \Delta_3=-400,
\]
as in the possibility (2) of Corollary~\ref{co:discriminant}. Thus, we conclude that the two quadrics are in transversal contact, as shown in Figure~\ref{fig:examplo-disciminants}. Notice that, in virtue of Remark~\ref{re:necessary-condition}, we do not need to check that the smallness condition is satisfied.

\begin{figure}
	\includegraphics[height=2cm]{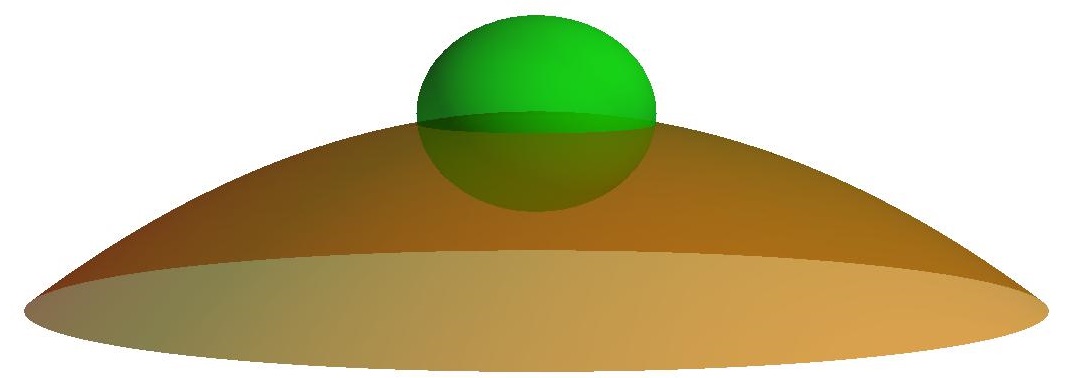}
	\caption{Contact is detected in terms of discriminants of the characteristic polynomial.}
	\label{fig:examplo-disciminants}
\end{figure}

\end{example}

\section{Relative positions of a small ellipsoid and another quadric}\label{sect:relative-positions}

When working with two objects in real world applications, sometimes it is important not only to detect contact, but also to know the relative position between them. Especially if the detection of contact involves solid objects whose border is modeled with surfaces, during a simulation it is a necessary task to detect when the smaller body is in the interior of the other. Thus, as a first step, Theorem~\ref{th:contact} allows to detect transversal contact between the surfaces, but if there is no contact, it would be desirable to know the relative position between the quadrics, this is, in which region of the space determined by the quadric $\mathcal{Q}$ is placed the small ellipsoid $\mathcal{E}$. This problem was already solved in \cite{BV-PS-SS-TT} for a small ellipsoid and an elliptic paraboloid.

We are considering quadrics $\mathcal{Q}$ which divide the projective space into two connected regions. Thus, considering the matrix $Q$ associated to the quadric $\mathcal{Q}$ and working in homogeneous coordinates $X=(x,y,z,1)$ as before, we distinguish regions $R_-$ and $R_+$ given by
\begin{equation}\label{eq:int-ext-regions}
R_-=\{(x,y,z)\in \mathbb{R}^3/ 	X^TQX\leq 0\} \text{ and } R_+=\{(x,y,z)\in \mathbb{R}^3 /X^TQX\geq 0\}.
\end{equation} 

Since we are working in affine space, these regions are not always connected, as it occurs with the hyperboloid of two sheets, where $R_-$ has two connected components (see Figure~\ref{fig:Interior-exterior}).
Moreover, in some cases, depending on the quadric $\mathcal{Q}$, we intuitively identify $R_-$ with the interior region and $R_+$ with the exterior one (for example, if we consider an ellipsoid). However this terminology is not so convenient for other quadrics, as the hyperbolic paraboloid. Therefore we will refer to these regions as $R_-$ and $R_+$. Also, note that we do not use strict inequalities in the definition of the two regions. Hence we allow tangent contact and still say that $\mathcal{E}$ belongs to $R_-$ or $R_+$. We emphasize that the intersection of  $R_-$ and $R_+$ is not empty, indeed, the two regions intersect in the points of the quadric surface.

The purpose of this section is to identify the relative position of a small ellipsoid $\mathcal{E}$ with respect to another quadric $\mathcal{Q}$. This relative position is considered from a topological viewpoint, so we are interested in detecting the region of the space divided by $\mathcal{Q}$ in which $\mathcal{E}$ is located. From Corollary~\ref{co:discriminant} we know how to detect contact in terms of the discriminants of the characteristic polynomial $\mathfrak{P}(\lambda)$. Thus, in what follows, we assume there is no contact, or just a tangent contact, so that either all points of $\mathcal{E}$ are located in $R_-$ or all of them belong to $R_+$. Our objective is to know in which of them are they in terms of the sign of the characteristic roots. Since the possible quadric surfaces $\mathcal{Q}$ have different features, we show in the next lemma how to deal with the hyperboloid of two sheets as a sample case. For other quadrics we proceed in an analogous way and we omit details in the interest of brevity (see Theorem~\ref{th:rel-pos} below).

\begin{lemma}\label{le:elip-hip-int-ext}
	Let $\mathcal{E}$ be a small ellipsoid and $\mathcal{H}$ a hyperboloid of two sheets. Then 
	\begin{itemize}
		\item $\mathcal{E}$ is placed in $R_-$ if and only if $\mathfrak{P}(\lambda)$ has four negative real roots.
		\item $\mathcal{E}$ is placed in $R_+$ if and only if $\mathfrak{P}(\lambda)$ has two negative and two positive real roots.
	\end{itemize}
\end{lemma}
\begin{proof}
First note that, if $\mathcal{E}$ is not completely within $R_-$ or $R_+$, then it is in transversal contact with $\mathcal{Q}$ and, by Theorem~\ref{th:contact}, there are non-real characteristic roots. Hence, let $\mathcal{E}$ be a small ellipsoid which is not in transversal contact with a hyperboloid of two sheets $\mathcal{H}$. Since the relative position is invariant under rigid moves, as are the roots of the characteristic polynomial (see \cite{wang-wang-kim}), we can locate $\mathcal{H}$ so that it is in standard form as in Table~\ref{table:smallness-condition-canonical-form} and its associated matrix is diagonal: $H=\operatorname{diag}\{\frac{1}{a^2},\frac{1}{b^2},-\frac{1}{c^2},1\}$. Notice also that by applying a rigid transformation the quadric still satisfy the smallness condition. 

We are going to place the center of the ellipsoid at two particular points and then argue using continuity to extend the result. 
Since both the relative position and the roots of the characteristic polynomial are invariant under scalings (see \cite{wang-wang-kim} for details), we firstly place the ellipsoid in $R_+$ so that the center is at $(0,0,0)$ (see Figure~\ref{fig:Interior-exterior}(b)) and secondly place it in $R_-$ so that it is tangent to the vertex of $\mathcal{H}$ (see Figure~\ref{fig:Interior-exterior}(c)). Now, appropriate scalings let us transform $\mathcal{E}$ into a sphere $\mathcal{S}$ of radius $1$. Note that the necessary scalings of the space that transform the ellipsoid into such a sphere also transform $\mathcal{H}$, but it is still of the given generic form.	
The equation of the sphere $S$ with center $(x_c, y_c, z_c)$ is 
	\begin{equation}\label{eq:sphere}
		(x-x_c)^2 + (y-y_c)^2 + (z-z_c)^2 = 1,
	\end{equation}
so the associated matrix is 
	\begin{equation}\label{eq:matrix-sphere}
		S=\left(   \begin{array}{cccc}
			1 & 0 & 0 & -x_c \\
			0 & 1 & 0 & -y_c  \\
			0 & 0 & 1 & -z_c \\
			-x_c & -y_c & -z_c & -1+x_c^2+y_c^2+z_c^2
		\end{array} \right ).
	\end{equation}
A direct calculation shows that the characteristic polynomial $\mathfrak{P}(\lambda)$ is given by
\begin{equation}\label{eq:charpoly-esfera-hip}
	\begin{array}{rcl}
	\mathfrak{P}(\lambda)\!\!&=\!\!&-\lambda^4+\left(\frac{x_c^2-1}{a^2}+\frac{y_c^2-1}{b^2}+\frac{1-z_c^2}{c^2}+1\right)\lambda^3\\
	\noalign{\smallskip}
	&& +\left(\frac{c^2 \left(a^2+b^2-1\right)+a^2+b^2-a^2
		b^2-\left(a^2+b^2\right) z_c^2+\left(c^2-a^2\right) y_c^2+\left(c^2-b^2\right) x_c^2}{a^2 b^2 c^2}\right)\lambda^2\\
	\noalign{\smallskip}
	&& -\frac{a^2+b^2-c^2+x_c^2+y_c^2+z_c^2-1}{a^2 b^2 c^2} \lambda-\frac{1}{a^2 b^2 c^2}.
	\end{array}
\end{equation}

Firstly, we consider the center of the sphere to be $(x_c,y_c,z_c)=(0,0,0)$, so $S$ is placed in $R_+$, and we see that  $a^2 b^2 c^2\mathfrak{P}(\lambda)=-(\lambda -1) \left(a^2 \lambda +1\right) \left(b^2 \lambda +1\right) \left(c^2 \lambda-1\right)$. Hence, there are 2 positive and 2 negative roots in this particular position. 

Secondly, we consider the center of the sphere to be $(x_c,y_c,z_c)=(0,0,c+1)$, so $S$ is placed in $R_-$, and we check that $a^2	b^2 c^2 \mathfrak{P}(\lambda)=-\left(a^2 \lambda +1\right) \left(b^2 \lambda +1\right) (c \lambda +1)^2$. Hence, there are 4 negative characteristic roots in this particular position.

Now, since $\mathfrak{P}(0)=-\frac{1}{a^2 b^2 c^2}\neq 0$, we have that $0$ is not a characteristic root. Given that the characteristic roots are real and vary continuously as we move the ellipsoid $\mathcal{E}$ and $0$ is not a root, the sign of the roots cannot change while we move $\mathcal{E}$ within $R_+$ or $R_-$. Thus,  because $R_+$ is connected, we conclude that there are 2 positive and 2 negative roots if $\mathcal{E}\in R_+$. Since $R_-$ has two connected components, we pass to the projective space (or simply repeat the previous calculation for $(x_c,y_c,z_c)=(0,0,-c-1)$), to conclude that there are 4 negative roots if $\mathcal{E}\in R_-$.
\end{proof}
 
\begin{figure}
	\includegraphics[height=3cm]{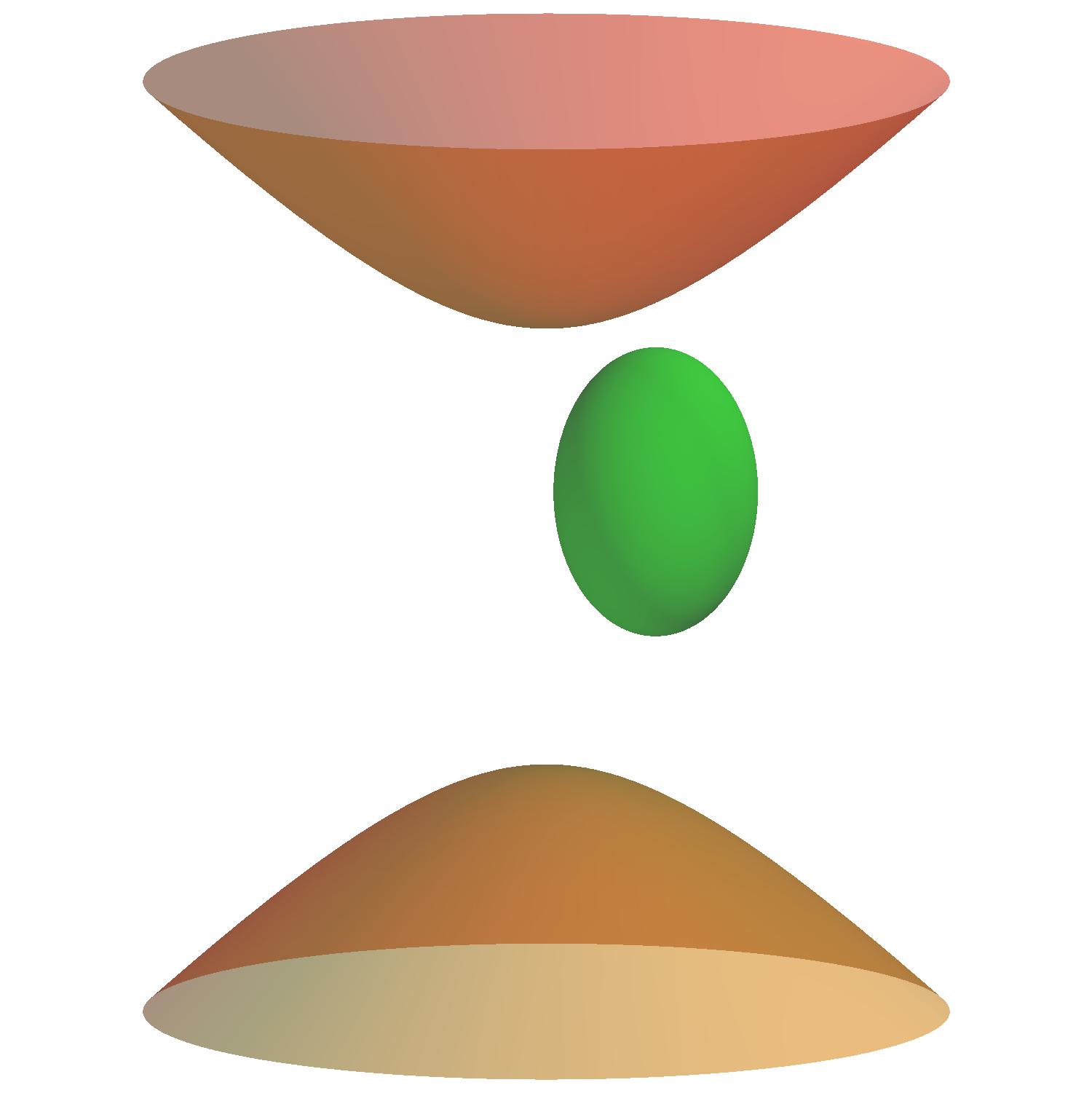}
	\includegraphics[height=3cm]{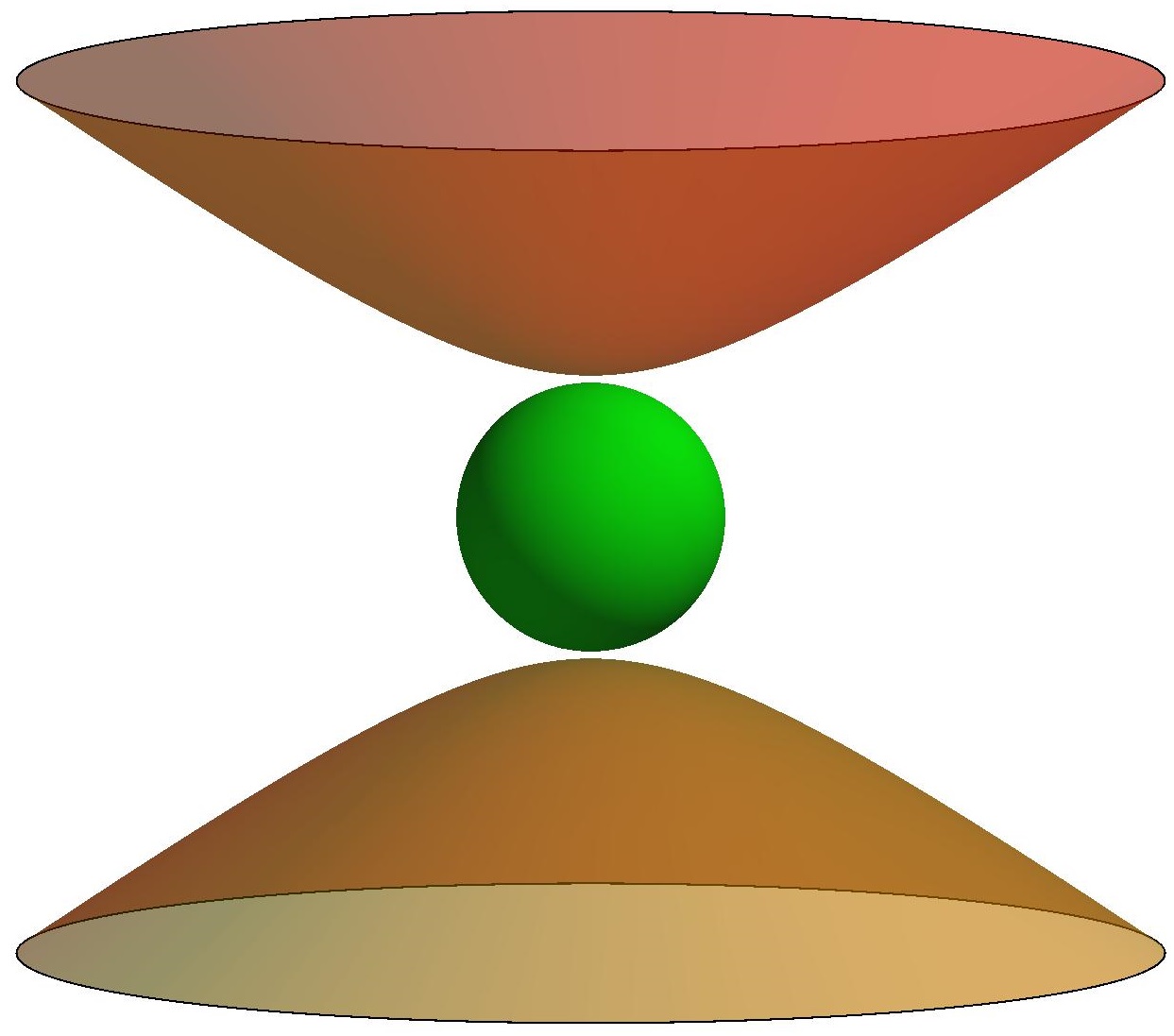}	\includegraphics[height=3cm]{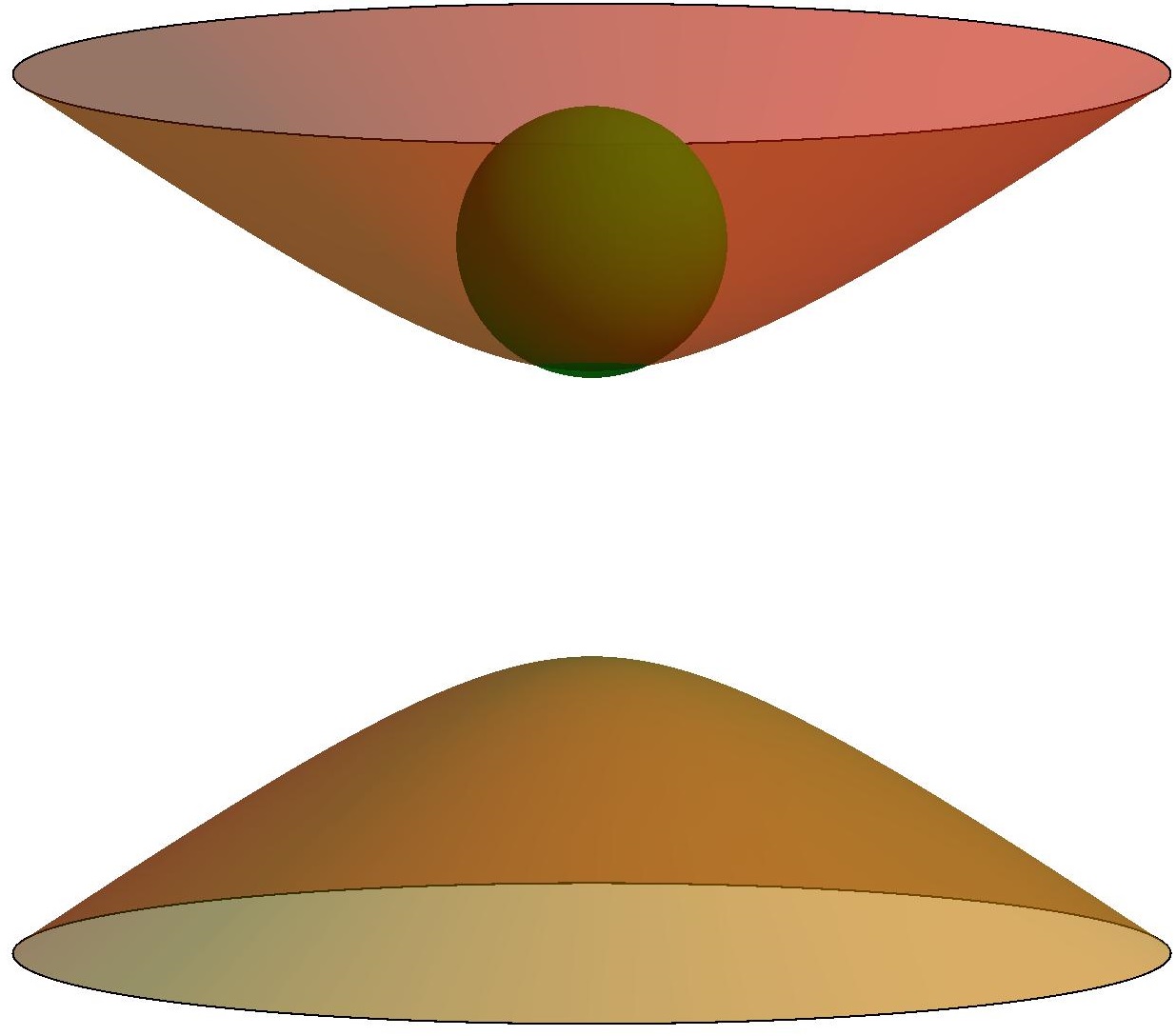}
		\qquad\qquad\qquad\qquad\qquad\qquad\qquad\qquad (a) \qquad\qquad\qquad\qquad\qquad (b) \qquad\qquad\qquad\qquad\qquad (c)
	\caption{The sign of the characteristic roots allows to distinguish interior and exterior cases.}\label{fig:Interior-exterior}
\end{figure}

Recall the form of the characteristic polynomial associated to the pair of quadrics $\mathcal{E}$ and $\mathcal{Q}$ given in expression \eqref{eq:char-poly-coef}. The signs of the roots are related to these coefficients, so that one can 
 distinguish the relative position just by checking the sign of the coefficients, without even computing the characteristic roots. This is shown in the following result.

\begin{theorem}\label{th:rel-pos}
	Let $\mathcal{E}$ be a small ellipsoid with respect to another quadric $\mathcal{Q}$. The relative position of the two quadrics for the non-contact (possibly tangent) cases is detected by the sign of the characteristic roots or, alternatively, the sign of the coefficients of $\mathfrak{P}(\lambda)$, as shown in Table~\ref{table:rel-pos}. 
\end{theorem}

\begin{table}[H]
	\begin{tabular}{|p{3.2cm}|p{2cm}|p{1.1cm}|l|}
		
		\hline
		
		\raisebox{-8pt}{\bf Quadric $\mathcal Q$}\rule[-0.4cm]{0cm}{1cm} & {\bfseries Relative position} & {\bfseries Sign of roots} & \raisebox{-8pt}{\makebox[4cm][c]{\bf Coefficients of $\mathfrak{P}$}}\\
		
		\hline\hline
		
		\multirow[c]{2}{*}[-0.5em]{ Ellipsoid} & \multirow[c]{3}{*}{ $X^TQX\leq 0$} & \multirow[c]{3}{*}{$----$} & $c_4, c_0 <0$ \\
		& & & No sing changes \\
		\multirow[c]{2}{3cm}{ Hyperboloid of two sheets} & & &$c_i\leq 0$ for $i=1,2,3$  \\\cline{2-4}
		
		& \multirow[c]{3}{*}{$X^TQX\geq 0$} & \multirow[c]{3}{*}{$--++$} & $c_4, c_0<0$ \\
		\multirow[c]{2}{*}{ Elliptic paraboloid} & & & $2$ sign changes  \\
		& & & $c_i>0$ for some $i=1,2,3$  \\
		
		\hline\hline
		
		\multirow[c]{3}{3cm}{Hyperboloid of 1 sheet} & \multirow[c]{3}{*}{$X^TQX\leq 0$} & \multirow[c]{3}{*}{$---+$} & $c_4<0, c_0>0$\\
		& & & No sign changes  \\ 
		& & & $c_i>0\Rightarrow c_j\geq 0 \forall j>i$ \\\cline{2-4}
		
		\multirow[c]{3}{3cm}{Hyperbolic paraboloid} & \multirow[c]{3}{*}{$X^TQX\geq 0$} & \multirow[c]{3}{*}{$-+++$} & $c_4<0, c_0>0$\\
		& & & $3$ sign changes \\	 
		& & & $c_i\leq 0, c_j>0$ for some $i<j$ \\
		
		\hline\hline
		
		\multirow[c]{3}{3cm}{Elliptic cylinder} & \multirow[c]{3}{*}{$X^TQX\leq 0$}  & \multirow[c]{3}{*}{$---0$} & $c_4<0, c_0=0$\\
		& & & No sign changes \\
		& & & $c_i\leq 0\forall i=1,2,3$ \\\cline{2-4}
		
		\multirow[c]{3}{3cm}{Parabolic cylinder} & \multirow[c]{3}{*}{$X^TQX\geq 0$} & \multirow[c]{3}{*}{$-0++$} & $c_4<0, c_0=0$ \\
		& & & $1$ sign change \\
		& & & $c_i>0$ for some $i=1,2,3$ \\
		
		\hline\hline
		
		\multirow[c]{6}{3cm}{Hyperbolic cylilnder} & \multirow[c]{3}{*}{$X^TQX\leq 0$} & \multirow[c]{3}{*}{$--0+$} & $c_4<0, c_0=0$ \\
		& & & $1$ sign change\\
		& & & $c_i<0\Rightarrow c_j\leq 0\forall j>i$\\\cline{2-4}
		
		& \multirow[c]{3}{*}{$X^TQX\geq 0$} & \multirow[c]{3}{*}{$0+++$} & $c_4<0, c_0=0$ \\
		& & & No sign changes\\
		& & & $(-1)^ic_i\leq 0 \forall i=1,2,3$\\
		
		\hline\hline
		
		\multirow[c]{4}{3cm}{Two parallel planes} & \multirow[c]{2}{*}{$X^TQX\leq 0$} & \multirow[c]{2}{*}{$--00$} & $c_i<0$ for $i=2,3,4$ \\
		& & & $c_0 = c_1 =0$\\\cline{2-4}
		
		& \multirow[c]{2}{*}{$X^TQX\geq 0$} & \multirow[c]{2}{*}{$00++$} & $(-1)^ic_i\leq 0$ for $i=2,3,4$\\
		& & & $c_0 = c_1 =0$\\
		
		\hline
		
	\end{tabular}
	\medskip
	\caption{Relative positions in terms of the sign of the roots and the sign of the coefficients of the characteristic polynomial.}
	\label{table:rel-pos}
\end{table}

\begin{proof}
The relation of the sign of the characteristic roots with the relative positions follows from Lemma~\ref{le:elip-hip-int-ext} if $\mathcal{Q}$ is a hyperboloid of two sheets. The argumentation in this lemma goes through if we change the hyperboloid of two sheets by any other non degenerate quadric $\mathcal{Q}$, simply by adapting the points we choose and the particular calculations. We do not include details in the interest of brevity.  If the quadric $\mathcal{Q}$ is degenerate instead, then $0$ is a characteristic root. In this case, one can reduce the dimension appropriately and a similar argument applies to obtain the remaining characterizations in Table~\ref{table:rel-pos}.  

In expression \eqref{eq:char-poly-coef}, we have that  $c_4=\det(\mathcal{E})$ and $c_0=\det (\mathcal{Q})$. Hence, some relations for the sign of the coefficients of $\mathfrak{P}$ are known a priori. Thus $c_4<0$ and, depending on the quadric $\mathcal{Q}$, $c_0>0$ or $c_0<0$ in the non-degenerate case, whereas $c_0=0$ for cylinders and $c_1=c_0=0$ for parallel planes.
Since the sign of the roots determines the relative position of the two quadrics, a direct application of the Descartes' rule of signs (see, for example, \cite{curtiss}) provides the relations in the signs given in Table~\ref{table:rel-pos}. 
\end{proof}

\section{Applications to detect contact between an ellipsoid and a combination of quadrics}\label{sect:combined-quadrics}

The results given in Theorem~\ref{th:contact}, Corollary~\ref{co:discriminant} and Theorem~\ref{th:rel-pos} provide a simple way of detecting contact and, moreover, the relative positions between an ellipsoid and another quadric. From these mathematical results efficient algorithms can be derived to be applied in real world applications. The simplicity of the analysis of the discriminant or the coefficients of the characteristic polynomial suits applications in many different contexts, so we are not going to specify a particular one here, but describe a general procedure of how to use the previous results. Here we consider static positions for the surfaces, however further analysis from the results in the previous sections can be carried out following the line of other works as \cite{etayo2006,X2011,Jia-Wang-Choi-Mourrain-Tu1} that deal with continuous collision detection.

\subsection{Using a plane to separate quadrics}\label{subsect:plane}
In virtue of Theorem~\ref{th:contact}, where a way of detecting contact is given for a variety of surfaces, one can combine different quadrics to create a model of an object. The combination of these quadrics can be done directly and one can even use one quadric to divide the space and differentiate two different zones where contact has to be checked with two different surfaces. Thus the detection of contact with an small ellipsoid is going to be carried out pairwise.

However, a simpler model can be done if we use a plane to separate zones. Since the intersection of a plane with a quadric surface is generically a conic curve, because of its manageability, it is a convenient way of creating a model of an object. Thus, one can divide space with a plane, check in which zone is the small ellipsoid and then detect contact with the corresponding quadric.

Notice that a plane can be represented by a matrix, in a similar way as a quadric, with no terms of order $2$. Hence, it makes sense to consider the characteristic polynomial \eqref{eq:char-poly} for an ellipsoid and a plane. The following result extends Theorem~\ref{th:contact} to the simpler case of an ellipsoid (with no restrictions in size or shape) and a plane. Moreover, it also extends Theorem~\ref{th:rel-pos} if we take into account that if $P$ denotes the matrix of the plane, then $R_-=\{(x,y,z)\in \mathbb{R}^3/ 	X^TPX\leq 0\}$ and $R_+=\{(x,y,z)\in \mathbb{R}^3 /X^TPX\geq 0\}$ as in \eqref{eq:int-ext-regions}.

\begin{theorem}\label{th:plane}
	Let $\mathcal{E}$ be an ellipsoid and $\mathcal{P}$ a plane. Then  $\mathcal{E}$ and $\mathcal{P}$ are in transversal contact if and only if the characteristic polynomial $\mathfrak{P}(\lambda)$ has two non-real roots.
	
	Furthermore, there is a double root which is zero and, if there is no transversal contact, the other two roots are positive if $\mathcal{E}\subset R_+$ whereas they are negative if $\mathcal{E}\subset R_-$. 
\end{theorem}
\begin{proof}
	We begin with a general ellipsoid $\mathcal{E}$ and an arbitrary plane $\mathcal{P}$. Since both the characteristic roots and the relative positions are invariant by rigid motion that preserve orientation, we can place the plane to be $xy$-plane. Moreover, by appropriate scalings, the ellipsoid can be transformed into a sphere and a new translation place it on the $z$-axis. Thus, without lost of generality, we consider a sphere $\mathcal{S}$ of radius $r$ with center at $(0,0,z_c)$ and the plane $\mathcal{P}_0$ with equation $z=0$. Their associated matrices are, respectively:
	\[
	S=
	\left(
	\begin{array}{cccc}
		\frac{1}{r^2} & 0 & 0 & 0 \\
		0 & \frac{1}{r^2} & 0 & 0 \\
		0 & 0 & \frac{1}{r^2} & -\frac{\text{zc}}{r^2} \\
		0 & 0 & -\frac{\text{zc}}{r^2} & \frac{z_c^2}{r^2}-1 \\
	\end{array}
	\right)
	\qquad \text{ and }\qquad
	P_0=\left(
	\begin{array}{cccc}
		0 & 0 & 0 & 0 \\
		0 & 0 & 0 & 0 \\
		0 & 0 & 0 & \frac{1}{2} \\
		0 & 0 & \frac{1}{2} & 0 \\
	\end{array}
	\right).
	\]
	Now, it is straightforward to check that 
	\begin{equation}\label{eq:char-poly-plane}
	\mathfrak{P}(\lambda)=-\frac{\lambda ^2 \left(r^2+4 \lambda  (\lambda -z_c)\right)}{4 r^6},
	\end{equation}
	so the characteristic roots are $\left\{0,0,\frac{1}{2} \left(z_c\pm\sqrt{z_c^2-r^2}\right)\right\}$. Since transversal contact occurs if and only if $|z_c|<r$, this is equivalent to the presence of two non-real roots.	
	
	Moreover, since $\sqrt{z_c^2-r^2}< z_c$, we have that $\frac{1}{2} \left(z_c\pm\sqrt{z_c^2-r^2}\right)>0$ if $z_c>r$ (so $\mathcal{E}\subset R_+=\{(x,y,z)\in \mathbb{R}^3 /z\geq 0\}$) and $\frac{1}{2} \left(z_c\pm\sqrt{z_c^2-r^2}\right)<0$ if $z_c<-r$ (so $\mathcal{E}\subset R_-=\{(x,y,z)\in \mathbb{R}^3 /z\leq 0\}$).
\end{proof}

As a consequence of Theorem~\ref{th:plane}, we can use the discriminant and the Descartes' rule of signs to detect the relative position between an ellipsoid and a plane as follows.

\begin{corollary}\label{co:plane}
	Let $\mathcal{E}$ be an ellipsoid and $\mathcal{P}$ a plane. Let $\mathfrak{P}(\lambda)=c_4 \lambda^4+c_3\lambda^3+c_2\lambda^2$ be the characteristic polynomial. Then
	\begin{enumerate}
		\item $\mathcal{E}$ and $\mathcal{P}$ are in transversal contact if and only if $\Delta_3<0$.
		\item If they are not in transversal contact, then
		\begin{itemize}
			\item $\mathcal{E}\subset R_+$ if and only if $c_3>0$.
			\item $\mathcal{E}\subset R_-$ if and only if $c_3<0$.
		\end{itemize} 
	\end{enumerate} 
\end{corollary}
\begin{proof}
From Theorem~\ref{th:plane} we have that $0$ is a double root, so we always have $\Delta_4=0$. Moreover, transversal contact corresponds to non-real roots so, from the characterization of roots given by a Complete Discrimination System (see, for example, \cite{yang}), it is characterized by $\Delta_3<0$. 

Now, from Theorem~\ref{th:plane}, we also have that the non-zero roots are positive if $\mathcal{E}\subset R_+$ and negative if $\mathcal{E}\subset R_-$. It follows from \eqref{eq:char-poly-plane} that $c_4<0$ and $c_2<0$ so, by the Descartes' rule of signs, two roots are positive if and only if $c_3>0$, whereas two roots are negative if and only if $c_3<0$.
\end{proof}

\begin{remark}\rm
Assuming we have a general plane $a x+b y+c z -d=0$ and an ellipsoid $\mathcal{E}$ with center at $(x_c,y_c,z_c)$ so that they are not in transversal contact, the sign of $a x_c+b y_c+c z_c -d$ determines if $\mathcal{E}$ belongs to $R_+$ or $R_-$. This is equivalent to check the sign of $c_3$ in Assertion (2) of Corollary~\ref{co:plane}, since $c_3=\frac{a x_c + b y_c + c z_c-d}{\alpha^2 \beta^2 \gamma^2}$ for $\mathcal{E}$ with parameters as in \eqref{eq:elipsoid-standar-form}.
\end{remark}

\subsection{Towards applications to real world models}

From the results in the previous section, we propose a simple algorithm to detect contact between an small ellipsoid $\mathcal{E}$ and an object which is modeled by a combination of quadrics. Using planes to divide space in several zones, quadrics can be combined in a simple way (see Figure~\ref{fig:arbore}). First, Theorem~\ref{th:plane} or Corollary~\ref{co:plane} are used to detect in which zone is the ellipsoid. Afterwards, depending on the zone, one can use Theorem~\ref{th:contact} or Corollary~\ref{co:discriminant} to detect contact. 


We use a simple model to describe the algorithm in more detail. In the most simple case there is one separating plane $\mathcal{P}$ dividing space into two zones (see the illustration of an example in Figure~\ref{fig:arbore}): Zone 1 ($R_+$) and Zone 2 ($R_-$). In Zone~1 there is a piece of the quadric surface $S_1$ and in Zone 2 there is a piece of the quadric surface $S_2$. In order to apply the results, the ellipsoid $\mathcal{E}$ shall be small with respect to both of them (this can be checked using Table~\ref{table:smallness-condition-general}). We work with the characteristic polynomials associated to $\mathcal{E}$ and $\mathcal{P}$ ($\mathfrak{P}_0(\lambda)$), to $\mathcal{E}$ and $S_1$ ($\mathfrak{P}_1(\lambda)$) and to $\mathcal{E}$ and $S_2$ ($\mathfrak{P}_2(\lambda)$). The algorithm to detect contact is divided into two steps and described as follows:
\begin{enumerate}
	\item {\it First step.} Using Corollary~\ref{co:plane} we detect if the ellipsoid $\mathcal{E}$ lies in Zone 1, Zone 2 or intersects the separating plane (Zone 0). 
	
	{\bf Data:} matrices associated to $\mathcal{E}$ and $\mathcal{P}$.
	
	{\bf Computations:} characteristic polynomial $\mathfrak{P}_0(\lambda)$. Discriminant $\Delta_3$ of $\mathfrak{P}_0(\lambda)$. Coefficient $c_3$ of $\mathfrak{P}_0(\lambda)$.
	
\begin{algorithm}[h]
	\eIf{$\Delta_3<0$}{$zone=Zone\,0$}{
	\eIf{$c_3>0$}{$zone=Zone\,1$}
	{$zone=Zone\, 2$}}
	\caption{Algorithm to detect the zone for the small ellipsoid}
\end{algorithm}

\item {\it Second step.} Depending on the zone the ellipsoid $\mathcal{E}$ is placed, we detect contact with Surface 1, with Surface 2 or with both of them.

	{\bf Data:} matrices associated to $\mathcal{E}$, $S_1$ and $S_2$.

{\bf Computations:} characteristic polynomials $\mathfrak{P}_1(\lambda)$ and $\mathfrak{P}_2(\lambda)$. Discriminants $\Delta_4^i$ and $\Delta_3^i$ of $\mathfrak{P}_i(\lambda)$, $i=1,2$. 

\begin{algorithm}
\uIf{zone=Zone 1}{
	\eIf{$\Delta_4^1<0$ $\vee$ ($\Delta_4^1=0$ $\wedge$ $\Delta_3^1<0$)}{contact between $\mathcal{E}$ and $S_1$}
	{no contact}}
\uElseIf{zone=Zone 2}{
	\eIf{$\Delta_4^2<0$ $\vee$ ($\Delta_4^2=0$ $\wedge$ $\Delta_3^2<0$)}{contact between $\mathcal{E}$ and $S_2$}
	{no contact}}
\Else{\eIf{$\Delta_4^1<0$ $\vee$ ($\Delta_4^1=0$ $\wedge$ 			$\Delta_3^1<0$) $\vee$ $\Delta_4^2<0$ $\vee$ 		($\Delta_4^2=0$ $\wedge$ $\Delta_3^2<0$)}{contact}
	{no contact}}
\caption{Algorithm to detect contact depending on the zone}
\end{algorithm}
\end{enumerate}

Note that the computations $\mathfrak{P}_1(\lambda)$ and $\mathfrak{P}_2(\lambda)$ (together with the associated discriminants $\Delta_4^i$ and $\Delta_3^i$) shall be computed only if needed, depending on the zone the ellipsoid is placed.

\begin{example}\rm
In order to illustrate theoretical results with an specific simple example we consider a bee and a tree as in Figure~\ref{fig:arbore}. The bee is modeled by an ellipsoid $\mathcal{E}:(x-3)^2+(y-3)^2+3(z-5.5)^2=0.1$, whereas the tree is modeled by an ellipsoid $S_1:x^2+y^2+3(z-8.36291)^2=20$ and a hyperboloid of one sheet $S_2:x^2+y^2-0.25 (z-3)^2=1$ separated by the plane $z=6$. 

{\it Smallness condition.} We check that the ellipsoid $\mathcal{E}$ is small with respect to $S_1$ and $S_2$ using directly Table~\ref{table:smallness-condition-general}. For $\mathcal{E}$ we have $\alpha=\beta=1$, $\gamma=\frac{1}{\sqrt{3}}$ and $\delta=\frac{1}{\sqrt{10}}$. For $S_1$ we have $a_1=b_1=1$, $c_1=\frac{1}{\sqrt{3}}$ and $d_1=\sqrt{20}$; and $a_2=b_2=1$, $c_2=2$ and $d_2=1$. Now, we check that
\[
\begin{array}{rcc}
\mathcal{E}-S_1: & d_1c_1=2.58199 \geq 0.316228= \delta \alpha,& \frac{a_1}{d_1c_1^2}=0.67082 \leq 1.82574=\frac{\gamma}{\delta \alpha^2},\\
\noalign{\smallskip}
\mathcal{E}-S_2: & d_2c_2=1 \geq 0.316228= \delta \alpha, & \frac{a_2}{d_2c_2^2}=0.25 \leq 1.82574=\frac{\gamma}{\delta \alpha^2}.
\end{array}
\]

{\it Detecting the zone.} First we consider the plane  to detect in which zone is the ellipsoid $\mathcal{E}$ located. We apply Algorithm 1. The characteristic polynomial is given by $\mathfrak{P}_0(\lambda)=\lambda ^2 \left(-0.3 \lambda ^2-1.5 \lambda -0.25\right)$, so $\Delta_3=0.121875>0$. Since the third degree coefficient is $c_3=-1.5<0$, we conclude that $\mathcal{E}$ lies in Zone 2.

{\it Detecting contact.} Once we know the ellipsoid $\mathcal{E}$ lies in Zone 2, we detect contact with the surface $S_2$. We apply Algorithm 2. The characteristic polynomial associated to $\mathcal{E}$ and $S_2$ is given by \[\mathfrak{P}_2(\lambda)=-0.3 \lambda ^4+45.7375 \lambda ^3+34.125 \lambda ^2-11.6625
\lambda +0.25,\] so $\Delta^2_4=6.90965\times 10^8>0$. We conclude that there is no contact with the surface. 

We can further detect that the ellipsoid $\mathcal{E}$ is out of the hyperboloid by checking that the coefficients of $\mathfrak{P}_2(\lambda)$ satisfy: $c_3>0$, $c_2>0$ and $c_1<0$, according to Table~\ref{table:rel-pos}.
\end{example}

\begin{figure}
	\includegraphics[width=6cm]{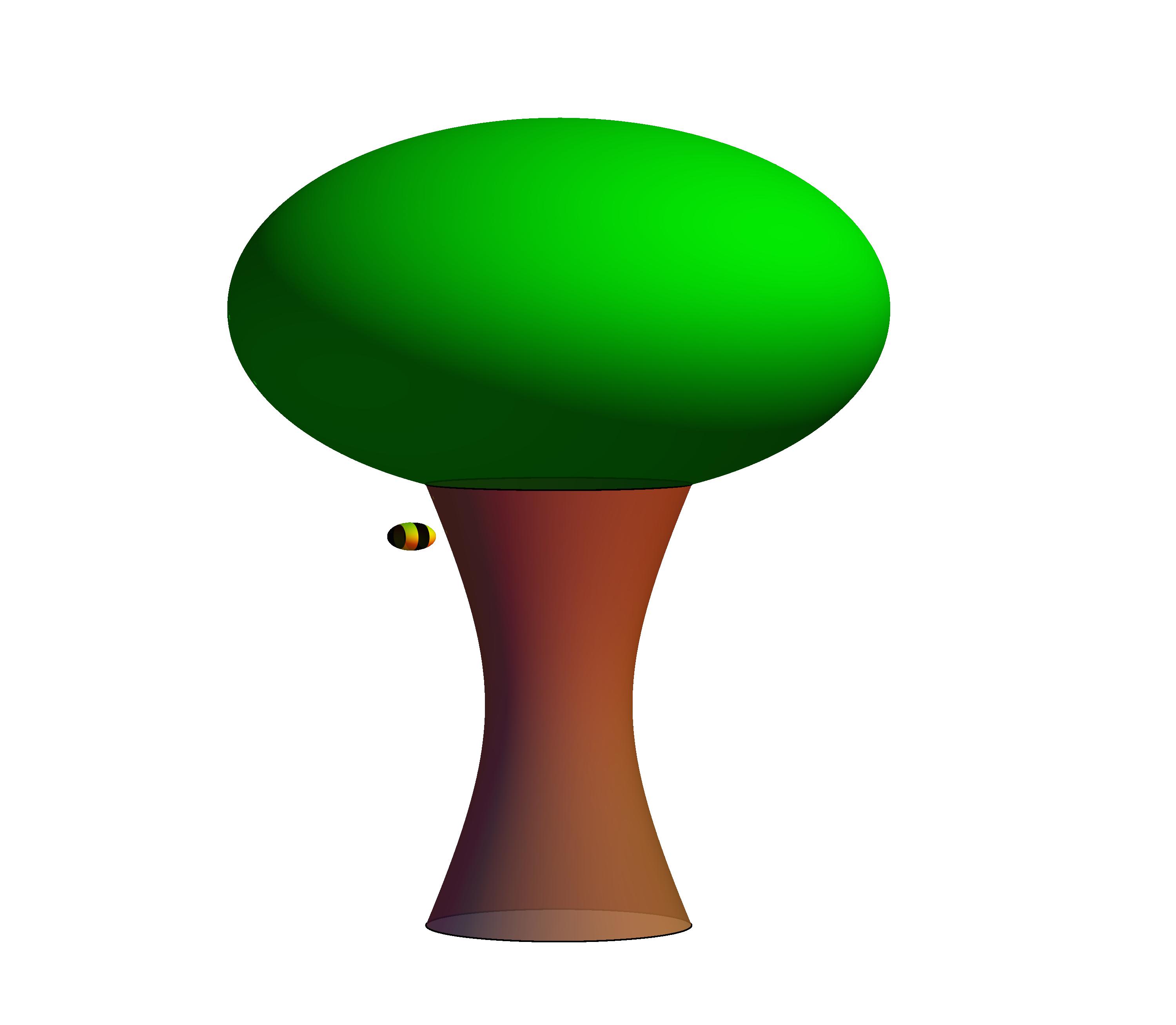}
		\includegraphics[width=6cm]{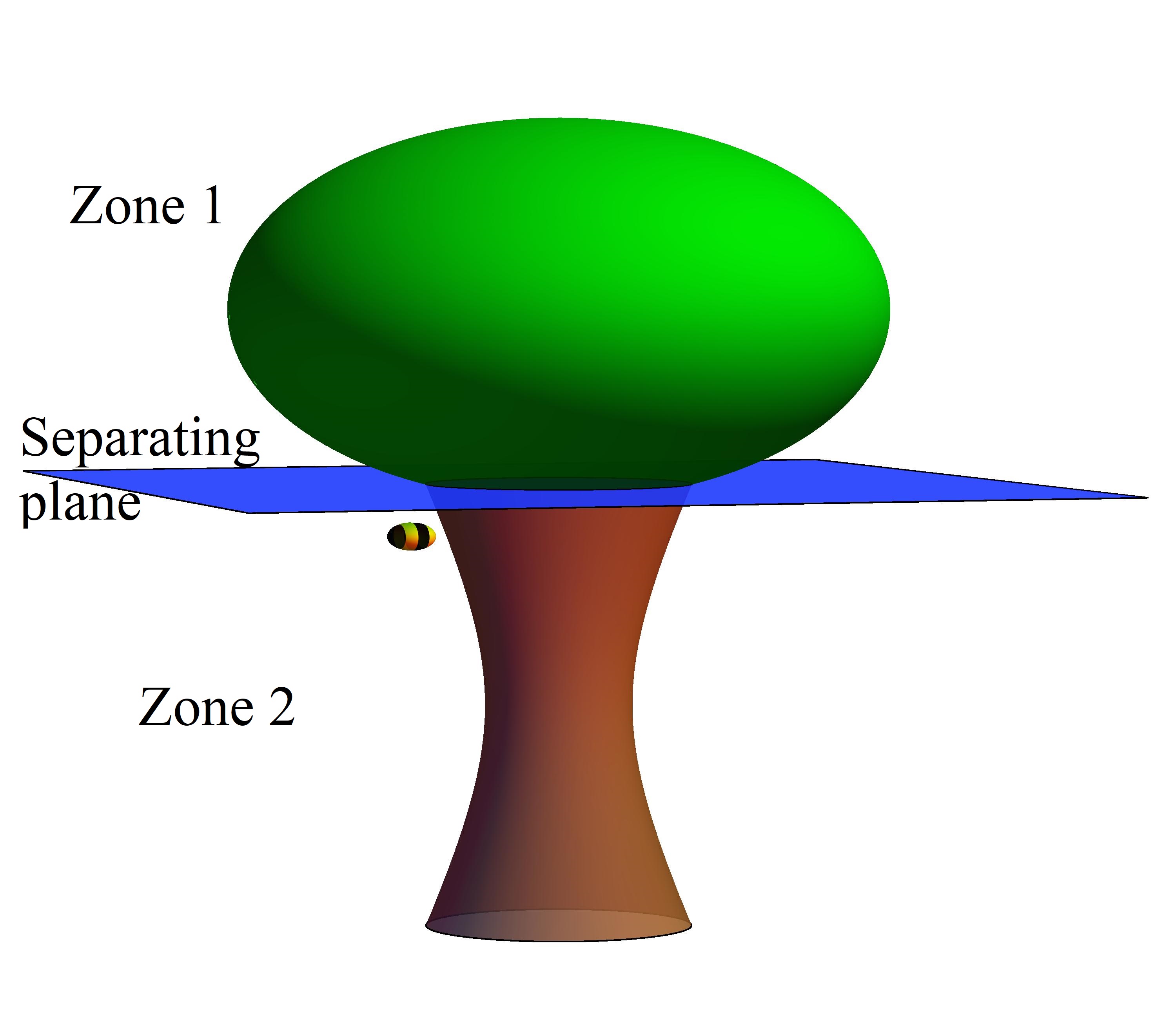}
	\caption{Simple model with a combination of two quadrics separated by a plane.}\label{fig:arbore}
\end{figure}

\section{Conclusions}\label{section:conclusions}
We have focused on the detection of transversal contact and on determining the relative position between two types of bounding volumes: an ellipsoid and all those surrounding objects that can be modeled by a combination of quadric surfaces.

An optimal hypothesis is introduced (see Definition~\ref{smallness-condition}) so that contact is detected by the existence of two non-real roots of the characteristic polynomial (Theorem~\ref{th:contact}). Moreover, that condition can be checked in terms of the parameters of the quadrics involved and contact can be detected by discriminants (Theorem~\ref{th:smallness-condition} and Corollary~\ref{co:discriminant}). Additionally, relative positions can be obtained from the sign of the coefficients of the characteristic polynomial (Theorem~\ref{th:rel-pos}).

We establish a broad theoretical framework that give rise to simple algorithms to detect collisions. Due to their efficiency and computational applicability, they can be used in a continuous time-varying positional contexts.

\end{document}